\documentclass[12pt,12pt,draftclsnofoot,onecolumn]{IEEEtran}
\usepackage[latin9]{inputenc}
\usepackage{color}
\usepackage{amsmath}
\usepackage{amsthm}
\usepackage{amssymb}
\usepackage{graphicx}

\makeatletter
\theoremstyle{plain}
	  \newtheorem{lem}{Lemma}
\theoremstyle{remark}
\newtheorem{rem}{Remark}
\theoremstyle{plain}
	  \newtheorem{prop}{Proposition}
\ifx\proof\undefined\
  \newenvironment{proof}[1][\proofname]{\par
    \normalfont\topsep6\p@\@plus6\p@\relax
    \trivlist
    \itemindent\parindent
    \item[\hskip\labelsep
          \scshape
      #1]\ignorespaces
  }{%
    \endtrivlist\@endpefalse
  }
  \providecommand{\proofname}{Proof}
\fi

\ifCLASSINFOpdf
\else
\fi
\usepackage{caption}
\usepackage{cite}
\renewcommand{\maketag@@@}[1]{\hbox{\m@th\normalsize\normalfont#1}}

\usepackage{etoolbox}
\usepackage{xstring}\usepackage{mfirstuc}\usepackage{textcase}\usepackage[acronym]{glossaries}
\newcommand{\newac}{\newacronym}
\newcommand{\ac}{\gls}

\newcommand{\acpl}{\glspl}

\makeglossaries
\newac{isac}{ISAC}{integrated sensing and communication}
\newac{an}{AN}{artificial noise}
\newac{csi}{CSI}{channel state information}
\newac{csit}{CSIT}{channel state information at transmitter}
\newac{fp}{FP}{fractional planning}
\newac{bme}{BME}{beampattern matching error}
\newac{ula}{ULA}{uniform linear array}
\newac{cu}{CU}{communication user}
\newac{uavs}{UAVs}{unmanned aerial vehicles}
\newac{sinr}{SINR}{signal-to-interference-plus-noise ratio}
\newac{sdr}{SDR}{semi-definite relaxation}
\newac{1d}{1D}{one-dimensional}
\newac{zf}{ZF}{zero-forcing}
\newac{aod}{AoD}{angle of departure}
\newac{los}{LoS}{line-of-sight}
\newac{nlos}{NLoS}{non-line-of-sight}
\newac{bs}{BS}{base station}
\newac{bss}{BSs}{base stations}
\newac{cus}{CUs}{communication users}
\newac{snr}{SNR}{signal-to-noise ratio}
\newac{evd}{EVD}{eigenvalue decomposition}
\newac{qsdp}{QSDP}{quadratic semidefinite programing}
\newac{sdp}{SDP}{semidefinite program}
\newac{dof}{DoF}{degrees-of-freedom}
\newac{soc}{SOC}{second order cone}
\newac{qos}{QoS}{quality-of-service}
\newac{lmi}{LMI}{linear matrix inequality}
\newac{bti}{BTI}{Bernstein-type inequality}
\newac{cscg}{CSCG}{circularly symmetric complex Gaussian}
\newac{crb}{CRB}{Cram{\'e}r-Rao bound}
\newac{sca}{SCA}{successive convex approximation}

\makeatother

\begin{document}
\title{Fundamental CRB-Rate Tradeoff in Multi-Antenna ISAC Systems with Information
Multicasting and Multi-Target Sensing}
\author{Zixiang Ren, Yunfei Peng, Xianxin Song, Yuan Fang, Ling Qiu, Liang
Liu, Derrick Wing Kwan Ng, and Jie Xu \thanks{Part of this paper has been presented at the 2022 IEEE Global Communications
Conference workshop, 4-8 December 2022, Rio de Janeiro, Brazil \cite{ren2022fundamental}.}\thanks{Z. Ren is with Key Laboratory of Wireless-Optical Communications,
Chinese Academy of Sciences, School of Information Science and Technology,
University of Science and Technology of China, and the Future Network
of Intelligence Institute (FNii), The Chinese University of Hong Kong
(Shenzhen), Shenzhen, China (e-mail: rzx66@mail.ustc.edu.cn).}\thanks{Y. Peng, X. Song, Y. Fang, and J. Xu are with the School of Science
and Engineering (SSE) and the FNii, The Chinese University of Hong
Kong (Shenzhen), Shenzhen, China (e-mail: 117010214@link.cuhk.edu.cn,
xianxinsong@link.cuhk.edu.cn, fangyuan@cuhk.edu.cn, xujie@cuhk.edu.cn). }\thanks{L. Qiu is with Key Laboratory of Wireless-Optical Communications,
Chinese Academy of Sciences, School of Information Science and Technology,
University of Science and Technology of China (e-mail: lqiu@ustc.edu.cn).}\thanks{L. Liu is with the Department of Electronic and Information Engineering,
The Hong Kong Polytechnic University, Hong Kong SAR, China (e-mails:
liang-eie.liu@polyu.edu.hk). }\thanks{D. W. K. Ng is with the University of New South Wales, Sydney, NSW
2052, Australia (e-mail: w.k.ng@unsw.edu.au).}\thanks{L. Qiu and J. Xu are the corresponding authors.}\vspace{-2cm}
}
\maketitle
\begin{abstract}
This paper investigates the performance tradeoff for a multi-antenna
integrated sensing and communication (ISAC) system with simultaneous
information multicasting and multi-target sensing, in which a multi-antenna
base station (BS) sends the common information messages to a set of
single-antenna communication users (CUs) and estimates the parameters
of multiple sensing targets based on the echo signals concurrently.
We consider two target sensing scenarios without and with prior target
knowledge at the BS, in which the BS is interested in estimating the
complete multi-target response matrix and the target reflection coefficients/angles,
respectively. First, we consider the capacity-achieving transmission
and characterize the fundamental tradeoff between the achievable rate
and the multi-target estimation Cramér-Rao bound (CRB) accordingly.
To this end, we design the optimal transmit signal covariance matrix
at the BS to minimize the estimation CRB for each of the two scenarios,
subject to the minimum multicast rate requirement and the maximum
transmit power constraint. It is shown that the optimal covariance
matrix consists of two parts for ISAC and dedicated sensing, respectively.
Next, we consider the transmit beamforming designs, in which the BS
sends one information beam together with multiple \textit{a-priori}
known dedicated sensing beams for effective ISAC and each CU can cancel
the interference caused by the sensing signals. By exploiting the
successive convex approximation (SCA) technique, we develop efficient
algorithms to obtain the joint information and sensing beamforming
solutions to the resultant rate-constrained CRB minimization problems.
Finally, we provide numerical results to validate the CRB-rate (C-R)
tradeoff achieved by our proposed designs, as compared to two benchmark
schemes, namely the isotropic transmission and the joint beamforming
without sensing interference cancellation. It is shown that the proposed
optimal transmit covariance solution achieves much better C-R performance
than the benchmark schemes and the proposed joint beamforming with
sensing interference cancellation performs close to the optimal transmit
covariance solution when the number of CUs is small. We also conduct
simulations to show the practical estimation performance achieved
by our proposed designs, by considering randomly generated information
signals and practical estimators.\vspace{-0.5cm}
\end{abstract}

\begin{IEEEkeywords}
Integrated sensing and communications (ISAC), multicast channel, multi-target
sensing, Cramér-Rao bound (CRB), transmit beamforming, optimization.\vspace{-0.5cm}
 
\end{IEEEkeywords}

\section{Introduction}

Future beyond fifth-generation (B5G) and sixth-generation (6G) wireless
networks are envisioned to support abundant intelligent applications
such as the metaverse, auto-driving, smart homes, and smart cities,
which demand ultra-high-capacity, ultra-low-latency, and ultra-high
reliable communications, as well as high-accuracy and high-resolution
sensing. Towards this end, \ac{isac} has emerged as a promising
technology for realizing B5G/6G, which enables the dual use of radio
signals and wireless network infrastructures for providing both sensing
and communication services simultaneously \cite{liu2021integrated}.
By seamlessly coordinating and integrating sensing and communications
\cite{LiuMasJ18}, ISAC is expected to significantly enhance both
communication and sensing performances, while improving the spectrum
utilization efficiency and the cost efficiency. As such, ISAC has
recently attracted tremendous research interests in both academia
and industry \cite{cui2021integrating}, motivating extensive studies
on ISAC from different research perspectives such as unified ISAC
waveforms \cite{liu2021integrated}, multi-antenna ISAC \cite{LiuMasJ18},
networked ISAC \cite{RahLusJ20,huang2022coordinated}, wideband ISAC
\cite{liu2021integrated}, and intelligent reflecting surface (IRS)-enabled
ISAC \cite{song2022intelligent,sankar2022beamforming}.

Among various enabling techniques, multiple-input multiple-output
(MIMO) has been recognized as an important solution to enhance the
performance of wireless networks, by equipping multiple antennas at
the wireless transceivers, e.g., \acpl{bs}. In particular, MIMO
can provide potential spatial multiplexing and diversity gains to
increase the communication rate and reliability \cite{heath2018foundations},
while offering spatial and waveform diversity gains to enhance the
sensing accuracy and resolution \cite{haimovich2007mimo,LiStoJ07}.
Indeed, how to optimize the transmit waveform and beamforming design
based on proper sensing and communication performance metrics is a
key technical challenge to be tackled for multi-antenna ISAC. On the
one hand, for communication, the \ac{sinr} and data rate are widely
adopted as performance measures. On the other hand, for sensing, various
performance measures may apply depending on the specific sensing tasks.
For instance, when the sensing estimation tasks for target detection
and tracking \cite{liu2022survey} are considered, the estimation
mutual information \cite{li2022framework}, radar sensing SINR \cite{zhang2022accelerating},
transmit beampattern \cite{LiuHuangNirJ20}, and \ac{crb} \cite{liu2021cramer}
are widely applied, among which the beampattern and \ac{crb}-based
designs are most promising. The basic idea of beampattern-based design
is to match the transmitted beampattern with a pre-defined sensing
beampattern such that the transmit signal beams are focused towards
the desired target directions and those leaked to the other undesired
directions are suppressed as possible. There have been extensive existing
works considering the waveform/beamforming design based on the transmit
beampattern. For example, the authors in \cite{LiuZhouJ18,LiuHuangNirJ20,hua2021optimal}
studied the downlink ISAC over broadcast channels by considering different
setups, in which the transmit beamforming was optimized to properly
balance the tradeoff between the transmit beampattern for sensing
and the received \ac{sinr} for communication. These designs were
then extended to the case of downlink non-orthogonal multiple access
(NOMA) in \cite{NOMAISAC}. However, aiming to mimic the transmit
beampattern may not be able to capture the quality of sensing directly
and thus may lead to sacrificed sensing performances.

In contrast to the transmit beampattern, the \ac{crb} has been
recognized as a practically relevant and analytically tractable metric
to directly measure the sensing performance limits for estimating
target parameters. In general, the CRB provides the lower bound of
variance for any unbiased estimators, which can be derived by exploiting
the Fisher information matrix to measure the amount of information
that the data provides about the parameter of interest \cite{liu2022survey,liu2021cramer,xiong2022flowing,Haocheng2022}.
As a result, directly optimizing the CRB is a viable new approach
that is suitable for designing ISAC systems, which not only facilitates
the characterization of the fundamental \ac{crb}-rate (C-R) performance
tradeoff, but also leads to potentially enhanced parameter estimation
performances. For instance, the authors in \cite{liu2021cramer} studied
the multiuser ISAC over a broadcast channel with one single-target,
in which the transmit beamforming was optimized to minimize the estimation
CRB, subject to the individual SINR (or equivalently rate) constraints
at multiple \acpl{cu}. Furthermore, \cite{xiong2022flowing} and
\cite{Haocheng2022} investigated the C-R tradeoff in a point-to-point
MIMO ISAC system with one CU and one target. In particular, the authors
in \cite{Haocheng2022} optimized the transmit covariance to characterize
the complete Pareto boundary of the C-R region for the MIMO ISAC system
in two specific scenarios with point and extended target models, respectively,
by considering the radar coherent processing interval to be sufficiently
long. As for \cite{xiong2022flowing}, the authors revealed the C-R
tradeoff for a more general case with finite radar coherent processing
intervals. In addition, \cite{yin2022rate} investigated the C-R tradeoff
for ISAC in multi-antenna broadcast channels, in which the emerging
rate-splitting multiple access (RSMA) technique was exploited to further
improve the ISAC performance. 

With recent advancements in webcast and content broadcasting applications,
it has become crucial for future B5G/6G wireless networks to effectively
support simultaneous data transmission to multiple users through multicast
channels \cite{jindal2006capacity,sidiropoulos2006transmit}. Additionally,
the ability to conduct multi-target sensing and tracking plays a vital
role in various multicast applications, such as surveillance, object
detection, and autonomous vehicles \cite{cui2021integrating}. Consequently,
supporting efficient ISAC over multi-antenna systems during such information
multicasting has become an increasingly important problem. This problem,
however, is particularly challenging. This is due to the fact that
the transmission principle for multicast channels significantly differs
from that for broadcast channels, thus rendering the prior ISAC designs
\cite{liu2021cramer,yin2022rate,Haocheng2022,xiong2022flowing} over
broadcast channels not applicable. To illustrate this, we consider
a multiple-input single-output (MISO) multicast channel as an example.
First, it has been demonstrated in \cite{jindal2006capacity} that
achieving the MISO multicast capacity typically requires a high-rank
transmit covariance matrix for transmitting common messages. This
stands in contrast to the MISO broadcast channel, where the capacity
is achievable by employing a rank-one covariance matrix for each user's
individual message along with dirty paper coding \cite{weingarten2006capacity}.
Next, it has been shown in \cite{sidiropoulos2006transmit} that in
MISO multicast channels, the transmit beamforming optimization problem
for minimizing the transmit power while ensuring individual users'
\ac{snr} requirements is an NP-hard problem. This is distinct from
the MISO broadcast channels, in which the transmit beamforming design
for SINR-constrained power minimization has been proven to be a convex
problem \cite{gershman2010convex,schubert2004solution}. Due to such
significant differences, new design approaches and principles are
needed for ISAC over multicast channels, thus motivating our investigation
in this work.

This paper investigates the multi-antenna ISAC over a multicast channel
with multi-target sensing, in which a multi-antenna \ac{bs} sends
common messages to a set of single-antenna \acpl{cu} and simultaneously
exploits the received echo signals to estimate the parameters of multiple
sensing targets. In particular, we consider two scenarios for multi-target
sensing without and with prior target knowledge at the BS, namely
Scenario I and Scenario II, respectively. In practice, Scenarios I
and II may correspond to the target detection and target tracking
stages for multi-target sensing, respectively.\footnote{In general, the BS often lacks prior knowledge of the targets before
transmission. Therefore, an initial target detection stage is required,
in which the ISAC signal transmission is designed to estimate the
complete multi-target response matrix for extracting multi-target
parameters. With the obtained target parameters at hand, a subsequent
target tracking stage is implemented, in which the ISAC signal transmission
is designed by leveraging such a piece of prior information to facilitate
the tracking of these targets or the estimation of their reflection
coefficients and angles.} The main results of our work are summarized as follows. 
\begin{itemize}
\item First, we consider the capacity-achieving transmission, based on which
we characterize the fundamental performance tradeoff between the estimation
CRB for multi-target sensing versus the channel capacity for information
multicasting in the two scenarios, which generalizes the target sensing
study compared with \cite{ren2022fundamental}. To this end, we minimize
the corresponding estimation CRB for each scenario, by optimizing
the transmit covariance matrix at the BS, subject to the minimum multicast
rate constraint and the maximum transmit power constraint. We obtain
the globally optimal solutions to the two rate-constrained CRB minimization
problems by applying advanced convex optimization techniques. For
both scenarios, the optimal transmit covariance is shown to consist
of two signal parts for ISAC and dedicated sensing, respectively.
\item Next, to facilitate the  implementation, we present new joint information
and sensing beamforming designs for efficient ISAC. In such designs,
the BS sends one common information beam together with multiple dedicated
a-priori known sensing beams such that each CU is capable of canceling
the interference caused by these sensing beams. For each of the two
scenarios, we develop a computationally efficient algorithm by adopting
the \ac{sca} technique to obtain a high-quality joint beamforming
solution for minimizing the multi-target estimation CRB while ensuring
the minimum multicast rate requirement. 
\item Finally, we provide numerical results to validate the C-R tradeoff
performance achieved by our proposed designs, as compared to two benchmark
schemes, namely the isotropic transmission and the joint beamforming
without sensing interference cancellation. It is shown that for each
scenario, the proposed designs significantly outperform the two benchmark
schemes and the proposed joint beamforming with sensing interference
cancellation achieves performance close to the upper bound achieved
by the optimal transmit covariance when the number of CUs is small.
We also perform simulations to assess the estimation performance using
randomly generated information signals and practical estimators. The
results clearly demonstrate that our CRB minimization designs significantly
reduce the root mean squared error (RMSE) of estimation compared to
conventional beampattern-based designs. This highlights the superior
performance and feasibility of our proposed designs in the practical
ISAC implementation and provides more comprehensive performance comparison
than \cite{ren2022fundamental}.
\end{itemize}

The remainder of this paper is organized as follows. Section II introduces
the multicast multi-target ISAC system model and formulates the multicast-rate-constrained
CRB minimization problems for Scenarios I and II when the BS does
not and does have the prior multi-target knowledge, respectively.
Sections III and IV develop the optimal solutions to the formulated
multicast-rate-constrained CRB minimization problems for Scenarios
I and II, respectively. Section V proposes  joint information and
sensing beamforming designs with sensing interference pre-cancellation
at the CUs. Section VI provides numerical results. Finally, Section
VII concludes this paper.

\textit{Notations}: Vectors and matrices are denoted by bold lower-
and upper-case letters, respectively. $\mathbb{C}^{N\times M}$ denotes
the space of $N\times M$ complex matrices. $\boldsymbol{I}$ and
$\boldsymbol{0}$ represent an identity matrix and an all-zero matrix
with appropriate dimensions, respectively. For a square matrix $\boldsymbol{A}$,
$\textrm{tr}(\boldsymbol{A})$ denotes its trace and $\boldsymbol{A}\succeq\boldsymbol{0}$
means that $\boldsymbol{A}$ is positive semi-definite. For a complex
arbitrary-size matrix $\boldsymbol{B}$, $\boldsymbol{B}[i,j]$, $\textrm{rank}(\boldsymbol{B})$,
$\boldsymbol{B}^{T}$, $\boldsymbol{B}^{H}$, and $\boldsymbol{B}^{c}$
denote its $(i,j)$-th element, rank, transpose, conjugate transpose,
and complex conjugate, respectively, and $\mathrm{vec}(\boldsymbol{B})$
denotes the vectorization of $\boldsymbol{B}$. For a vector $\boldsymbol{a}$,
$\boldsymbol{a}[i]$ denotes its $i$-th element. $\mathbb{E}(\cdot)$
denotes the statistical expectation. $\|\cdot\|$ denotes the Euclidean
norm of a vector. $|\cdot|$, $\mathrm{Re}(\cdot)$, and $\mathrm{Im}(\cdot)$
denote the absolute value, the real component, and the imaginary component
of a complex entry. $\mathcal{CN}(\boldsymbol{x},\boldsymbol{Y})$
denotes a \ac{cscg} random vector with mean vector $\boldsymbol{x}$
and covariance matrix $\boldsymbol{Y}$. $\boldsymbol{A}\otimes\boldsymbol{B}$
and $\boldsymbol{A}\odot\boldsymbol{B}$ represent the Kronecker product
and Hadamard product of two matrices $\boldsymbol{A}$ and $\boldsymbol{B}$,
respectively.\vspace{-0.5cm}

\section{System Model and Problem Formulation}

\begin{figure}
\vspace{-0.8cm}
\includegraphics[scale=0.35]{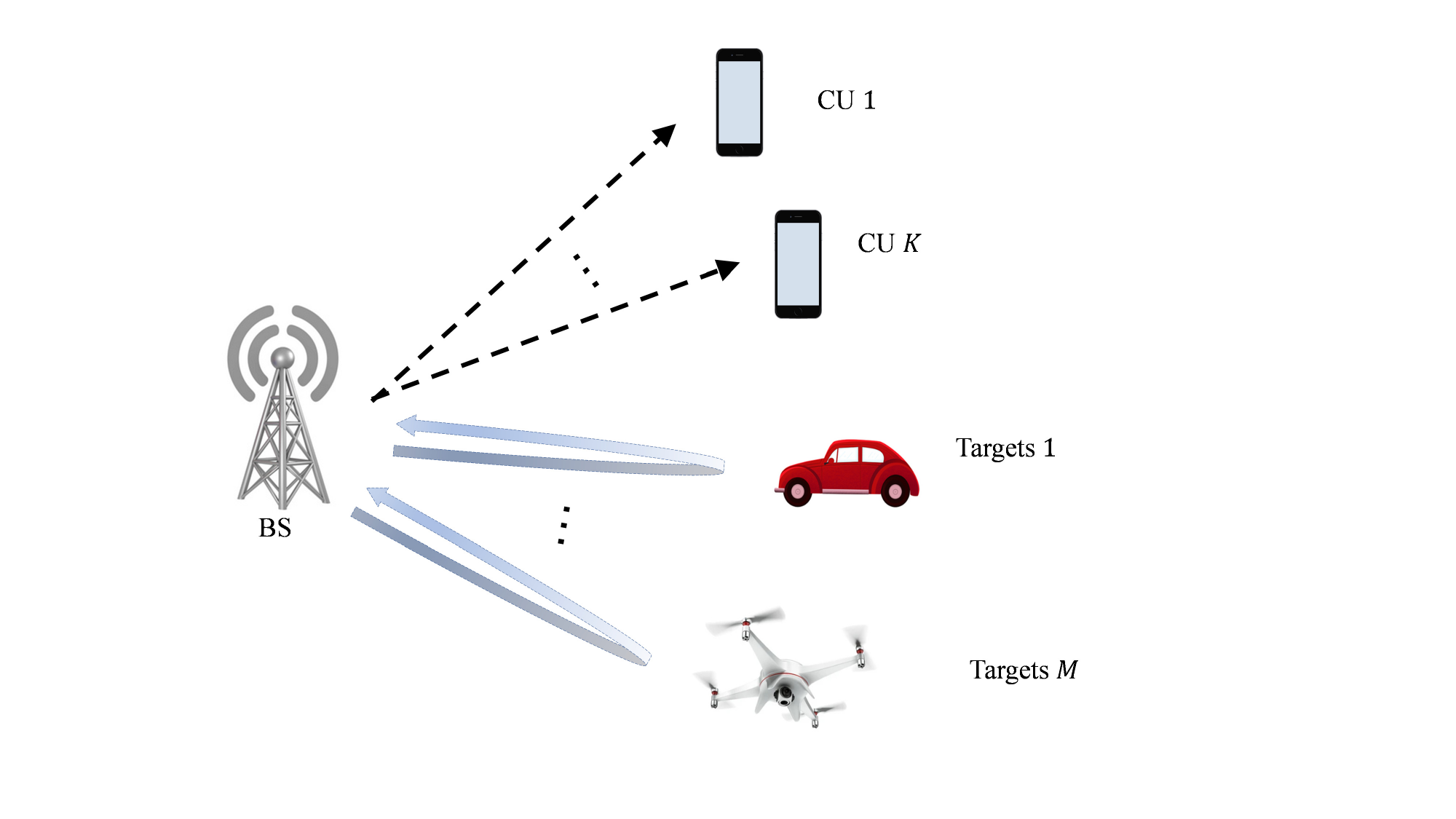}\centering\vspace{-0.8cm}
\caption{\label{fig:1-1}Illustration of the multi-antenna ISAC system for
simultaneous information multicasting and multi-target sensing. }
\vspace{-0.8cm}
\end{figure}

As shown in Fig. \ref{fig:1-1}, we consider a multi-antenna ISAC
system for information multicasting and multi-target sensing, which
consists of a multi-antenna \ac{bs} equipped with $N_{t}$ transmit
antennas and $N_{r}$ receive antennas, $K$ single-antenna \acpl{cu},
and $M$ sensing targets. Let $\mathcal{K}\overset{\bigtriangleup}{=}\{1,2,\ldots,K\}$
denote the set of CUs and $\mathcal{M}\overset{\bigtriangleup}{=}\{1,2,\ldots,M\}$
denote that of sensing targets. The BS sends common information signals
to all the CUs and exploits the received echo signals to sense the
$M$ targets simultaneously. In particular, we focus on the ISAC period
consisting of $L$ symbols. In this paper, we focus on a quasi-static
environment, where the communication and sensing channels are assumed
to remain invariant during the ISAC period, similar to \cite{liu2021cramer,li2007range}.
Let $\boldsymbol{x}(n)\in\mathbb{C}^{N_{t}\times1}$ denote the transmitted
ISAC signal by the BS in symbol $n\in\mathcal{L}\triangleq\{1,2,\ldots,L\}$.
For achieving the multicast capacity, we consider that the transmitted
signal $\boldsymbol{x}(n)$ is a \ac{cscg} random vector with zero
mean and covariance $\mathbb{E}[\boldsymbol{x}(n)\boldsymbol{x}^{H}(n)]=\boldsymbol{S}_{x}\succeq\boldsymbol{0}$,
i.e., $\boldsymbol{x}(n)\sim\mathcal{CN}(\boldsymbol{0},\boldsymbol{S}_{x}),\forall n$.
Supposing that the BS is subject to a maximum transmit power budget
$P$, we have \vspace{-0.2cm}
\begin{equation}
\mathrm{tr}(\boldsymbol{S}_{x})\leq P.\label{eq:Average power}
\end{equation}
\vspace{-1cm}

\subsection{Multicasting Communication Model}

First, we consider the multicast channel for communication. Let $\boldsymbol{h}_{k}\in\mathbb{C}^{N_{t}\times1}$
denote the channel vector from the BS to CU $k\in\mathcal{K}$. The
received signal at the receiver of \ac{cu} $k\in\mathcal{K}$ is
given by \cite{jindal2006capacity}\vspace{-0.5cm}
\begin{equation}
y_{k}(n)=\boldsymbol{h}_{k}^{H}\boldsymbol{x}(n)+z_{k}(n),\forall n\in\mathcal{L},\label{eq:Received com signal}
\end{equation}
\vspace{-0.2cm}
where $z_{k}(n)$ denotes the noise at the receiver of CU $k$ that
is a CSCG random variable with zero mean and variance $\sigma^{2}$,
i.e., $z_{k}(n)\sim\mathcal{CN}(0,\sigma^{2}),\forall k\in\mathcal{K}$.
Based on the received signal in \eqref{eq:Received com signal}, the
received \ac{snr} at \ac{cu} $k\in\mathcal{K}$ is\footnote{Since all the receivers require the same common message from one single
BS, there is no multi-user interference. On the other hand, the co-channel
interference from other radio transmitters in the environment can
be included in $z_{k}(n)$ such that $\sigma^{2}$ captures the impacts
caused by background interference.}\vspace{-0.2cm}
\begin{equation}
\gamma_{k}(\boldsymbol{S}_{x})=\mathbb{E}\bigg(\frac{\big|\boldsymbol{h}_{k}^{H}\boldsymbol{x}(n)\big|^{2}}{\big|z_{k}(n)\big|^{2}}\bigg)=\frac{\boldsymbol{h}_{k}^{H}\boldsymbol{S}_{x}\boldsymbol{h}_{k}}{\sigma^{2}}.\label{eq:SNR for com}
\end{equation}
\vspace{-0.2cm}
Accordingly, the achievable rate of the multicast channel with a given
transmit covariance $\boldsymbol{S}_{x}$ is given by \cite{jindal2006capacity,sidiropoulos2006transmit}\vspace{-0.2cm}
\begin{equation}
R(\boldsymbol{S}_{x})=\underset{k\in\mathcal{K}}{\min}\Big\{\log_{2}\Big(1+\frac{\boldsymbol{h}_{k}^{H}\boldsymbol{S}_{x}\boldsymbol{h}_{k}}{\sigma^{2}}\Big)\Big\}.\label{eq:rate}
\end{equation}
\vspace{-0.2cm}
As a result, the capacity of the multicast channel is defined as the
maximum achievable rate, given by \cite{jindal2006capacity,sidiropoulos2006transmit}\vspace{-0.2cm}
\begin{equation}
C_{0}=\underset{\boldsymbol{S}_{x}\succeq\boldsymbol{0},\mathrm{tr}(\boldsymbol{S}_{x})\leq P}{\max}R(\boldsymbol{S}_{x}).
\end{equation}
\vspace{-0.2cm}

\subsection{Multi-target Sensing Model}

Next, we consider the multi-target radar sensing, in which the BS
adopts its $N_{r}$ receive antennas for receiving the echo signals
to estimate the parameters of the $M$ sensing targets. Suppose that
these $M$ targets are located in the same range bin of the far-field
region of the BS. In this case, the received target echo signals at
the BS in symbol $n\in\mathcal{L}$ is expressed as\footnote{Similar as in prior works on ISAC and radar sensing \cite{liu2021cramer,Haocheng2022},
we ignore the self-interference due to the signal transmission in
(5), as the sensing receiver can handle such an issue properly \cite{zeng2018joint}.
This is different from the communication receivers, for which the
self-interference might be a severe issue. Furthermore, similar to
prior works on MIMO radar sensing \cite{LiStoJ07,li2007range} and
MIMO ISAC \cite{liu2021integrated}, we consider the estimation of
target angles and complex coefficients based on the echo signals in
(6), in which the Doppler effects are omitted.}\vspace{-0.2cm}
\begin{equation}
\boldsymbol{y}(n)=\sum_{m=1}^{M}\beta_{m}\boldsymbol{a}_{r}^{c}(\theta_{m})\boldsymbol{a}_{t}^{H}(\theta_{m})\boldsymbol{x}(n)+\boldsymbol{\boldsymbol{z}}(n),\label{eq:Received signal for sensing}
\end{equation}
where $\beta_{m}\in\mathbb{C}$ denotes the complex coefficient of
target $m\in\mathcal{M}$ whose amplitude captures the round-trip
path loss and is proportional to its radar cross section (RCS) \cite{LiuMasJ202,li2007range},
$\theta_{m}$ denotes the angle of departure (AoD)/angle of arrival
(AoA) of target $m\in\mathcal{M}$, $\boldsymbol{a}_{r}(\theta)\in\mathbb{C}^{N_{r}\times1}$
and $\boldsymbol{a}_{t}(\theta)\in\mathbb{C}^{N_{t}\times1}$ denote
the receive and transmit steering vectors with angle $\theta$, respectively,
and $\boldsymbol{z}(n)\sim\mathcal{CN}(\boldsymbol{0},\sigma_{r}^{2}\boldsymbol{I})$
is the noise at the BS receiver with $\sigma_{r}^{2}$ being the noise
power at each receive antenna. Let $\boldsymbol{X}=[\boldsymbol{x}(1),\dots,\boldsymbol{x}(L)]$
denote the transmitted signal over the $L$ symbols. By assuming that
$L$ is sufficiently large, the sample covariance matrix of $\boldsymbol{X}$
can be approximated as the statistical covariance matrix\footnote{We will show that the approximation is sufficiently accurate for practical
values of $L$ (e.g., $L\ge64$) in the numerical results in Section
VI. } $\boldsymbol{S}_{x}$, i.e.,\vspace{-0.2cm}
\begin{equation}
\frac{1}{L}\boldsymbol{X}\boldsymbol{X}^{H}\approx\boldsymbol{S}_{x}.\label{eqn:app}
\end{equation}
\vspace{-0.2cm}
To facilitate the derivation of the \ac{crb} matrix for estimating
the parameters of targets, we rewrite \eqref{eq:Received signal for sensing}
as\vspace{-0.2cm}
\begin{equation}
\boldsymbol{Y}=\boldsymbol{A}_{r}^{c}\boldsymbol{B}\boldsymbol{A}_{t}^{H}\boldsymbol{X}+\boldsymbol{Z},\label{eq:Total received signal}
\end{equation}
\vspace{-0.2cm}
where $\boldsymbol{Y}=[\boldsymbol{y}(1),...,\boldsymbol{y}(L)]$,
$\boldsymbol{Z}=[\boldsymbol{z}(1),...,\boldsymbol{z}(L)]$, $\boldsymbol{A}_{r}=[\boldsymbol{a}_{r}(\theta_{1})$,
$\boldsymbol{a}_{r}(\theta_{2}),\dots,\boldsymbol{a}_{r}(\theta_{M})]$,
$\boldsymbol{A}_{t}=[\boldsymbol{a}_{t}(\theta_{1}),\boldsymbol{a}_{t}(\theta_{2}),\dots,\boldsymbol{a}_{t}(\theta_{M})]$,
$\boldsymbol{\theta}=[\theta_{1},\theta_{2},\dots,\theta_{M}]^{T},\boldsymbol{\beta}=[\beta_{1},\beta_{2},\dots,\beta_{M}]^{T}$,
and $\boldsymbol{B}=\mathrm{diag}(\boldsymbol{\beta})$. Let $\boldsymbol{G}=\boldsymbol{A}_{r}^{c}\boldsymbol{B}\boldsymbol{A}_{t}^{H}$
denote the multi-target response matrix. Accordingly, we vectorize
\eqref{eq:Total received signal} as\vspace{-0.2cm}
\begin{equation}
\tilde{\boldsymbol{y}}=(\boldsymbol{X}^{T}\otimes\boldsymbol{I}_{N_{r}})\tilde{\boldsymbol{g}}+\tilde{\boldsymbol{z}},\label{eq:vector form}
\end{equation}
\vspace{-0.2cm}
where $\tilde{\boldsymbol{y}}=\mathrm{vec}(\boldsymbol{Y})\in\mathbb{C}^{N_{r}L\times1}$,
$\tilde{\boldsymbol{g}}=\mathrm{vec}(\boldsymbol{G})\in\mathbb{C}^{N_{r}N_{t}\times1},$
and $\tilde{\boldsymbol{z}}=\mathrm{vec}(\boldsymbol{Z})\in\mathbb{C}^{N_{r}L\times1}$.
It follows from \eqref{eq:vector form} that the received signal vector
$\tilde{\boldsymbol{r}}$ is a CSCG random vector with mean $(\boldsymbol{X}^{T}\otimes\boldsymbol{I})\tilde{\boldsymbol{g}}$
and covariance $\sigma_{r}^{2}\boldsymbol{I}$, i.e., $\tilde{\boldsymbol{y}}\sim\mathcal{CN}\big((\boldsymbol{X}^{T}\otimes\boldsymbol{I})\tilde{\boldsymbol{g}},\sigma_{r}^{2}\boldsymbol{I}\big)$.

In particular, we consider two different sensing scenarios, namely
Scenario I and Scenario II, in which the BS does not and does have
\textit{a-prior} knowledge about the multiple targets such that the
BS aims to estimate the complete multi-target response matrix $\boldsymbol{G}=\boldsymbol{A}_{r}^{c}\boldsymbol{B}\boldsymbol{A}_{t}^{H}$
and the targets' coefficients/angles (i.e., $\{\beta_{m}\}_{m=1}^{M}$
and $\{\theta_{m}\}_{m=1}^{M}$), respectively. In practice, Scenario
I may correspond to the target detection phase, in which the BS tries
to identify the number of targets and obtain their parameters based
on matrix $\boldsymbol{G}$ via spatial spectrum algorithms such as
multiple signal classification (MUSIC). On the other hand, Scenario
II may correspond to the target tracking phases, in which the BS needs
to track the $M$ targets by estimating their coefficients/angles,
with their initial parameters known in advance \cite{huang2022joint,LiuMasJ20},
respectively.

\subsubsection{CRB for Estimating $\boldsymbol{G}$ in Scenario I}

First, we consider Scenario I, in which the BS lacks prior knowledge
of the targets in the initial target detection stage \cite{StoPETLiJ07,liu2021integrated}.
In this scenario, the BS is interested in estimating the complete
multi-target response matrix $\boldsymbol{G}=\boldsymbol{A}_{r}^{c}\boldsymbol{B}\boldsymbol{A}_{t}^{H}$
based on the received signal $\boldsymbol{Y}$ in \eqref{eq:Total received signal},
or equivalently estimating $\tilde{\boldsymbol{g}}$ based on $\tilde{\boldsymbol{r}}$
in \eqref{eq:vector form}. After obtaining the estimate of $\boldsymbol{G}$,
the BS can further implement spatial spectrum algorithms such as MUSIC\cite{schmidt1986multiple}
to further extract the angle information of these targets. In this
scenario, the BS needs to estimate $N_{r}N_{t}$ complex parameters
in $\tilde{\boldsymbol{g}}$ or $\boldsymbol{G}$. Based on the complex
linear model in \eqref{eq:vector form}, it has been established in
\cite{kay1993fundamentals} that the \ac{crb} matrix for estimating
$\tilde{\boldsymbol{g}}$ is\vspace{-0.2cm}
\begin{align}
\boldsymbol{C}_{1}\negthickspace= & \big((\boldsymbol{X}^{T}\otimes\boldsymbol{I}_{N_{r}})^{H}(\sigma_{r}^{2}\boldsymbol{I})^{-1}(\boldsymbol{X}^{T}\otimes\boldsymbol{I}_{N_{r}})\big)^{-1}\negthickspace\negthickspace=\sigma_{r}^{2}\big(\boldsymbol{X}^{c}\boldsymbol{X}^{T}\otimes\boldsymbol{I}_{N_{r}}\big)^{-1}\negthickspace\overset{(\mathrm{a})}{=}\frac{\sigma_{r}^{2}}{L}\big(\boldsymbol{S}_{x}^{T}\otimes\boldsymbol{I}_{N_{r}}\big)^{-1},\label{eq:CRB1}
\end{align}
\vspace{-0.2cm}
where (a) follows from the approximation in \eqref{eqn:app}. For
facilitating the sensing performance optimization, scalar functions
of the CRB matrix are normally adopted as the performance metrics,
some examples include trace, determinant, and minimum eigenvalue \eqref{eq:CRB1}.
In this paper, we adopt its trace as the scalar performance metric,
i.e.\vspace{-0.2cm}
\begin{equation}
\mathrm{CRB}_{1}(\boldsymbol{S}_{x})=\frac{N_{r}\sigma_{r}^{2}}{L}\mathrm{tr}(\boldsymbol{S}_{x}^{-1}).\label{eq:CRB}
\end{equation}
\vspace{-0.2cm}
Intuitively, $\mathrm{CRB}_{1}(\boldsymbol{S}_{x})$ corresponds to
the sum of the CRBs for estimating the elements in $\boldsymbol{G}$.

\subsubsection{CRB for Estimating $\boldsymbol{\theta}$ and $\boldsymbol{\beta}$
in Scenario II}

Next, we consider Scenario II, in which the BS has a priori knowledge
of the multiple targets (e.g., from the target detection stage), and
thus is interested in estimating the target coefficients $\boldsymbol{\beta}$
and angles $\boldsymbol{\theta}$ as unknown parameters. This may
correspond to the target tracking stage in practice \cite{StoPETLiJ07,liu2021integrated}.
Let $\boldsymbol{\beta}_{R}=[\mathrm{Re}(\beta_{1}),\mathrm{Re(}\beta_{2}),\dots,\mathrm{Re(}\beta_{M})]^{T}$
and $\boldsymbol{\beta}_{I}=[\mathrm{Im}(\beta_{1}),\mathrm{Im(}\beta_{2}),\dots,\mathrm{Im(}\beta_{M})]^{T}$
denote the real and imaginary parts of $\boldsymbol{\beta}$. Accordingly,
we have a total of $3M$ real parameters to be estimated, given by
$\boldsymbol{\xi}=[\boldsymbol{\theta}^{T},\boldsymbol{\text{\ensuremath{\beta}}}_{R}^{T},\boldsymbol{\beta}_{I}^{T}]$$^{T}$.
To facilitate the CRB derivation, we have the following lemma \cite{kay1993fundamentals}. 
\begin{lem}
\label{lem:1-1}\textup{\cite{kay1993fundamentals} Consider the estimation
of real-valued parameters $\boldsymbol{\varsigma}\in\mathbb{R}^{T\times1}$
based on the CSCG data vector $\tilde{\boldsymbol{v}}\in\mathbb{C}^{N\times1}$,
i.e., $\tilde{\boldsymbol{v}}\sim\mathcal{CN}(\tilde{\boldsymbol{\mu}}(\boldsymbol{\varsigma}),\boldsymbol{C}_{\tilde{\boldsymbol{v}}})$,
where the covariance $\boldsymbol{C}_{\tilde{\boldsymbol{v}}}$ is
independent from parameters $\boldsymbol{\varsigma}$. The Fisher
information matrix (FIM) for estimating $\boldsymbol{\varsigma}$
is given by $\boldsymbol{F}_{\boldsymbol{\varsigma}}$, with}\vspace{-0.2cm}
\textup{
\begin{equation}
\boldsymbol{F}_{\boldsymbol{\varsigma}}[i,j]=2\mathrm{Re}\Big(\frac{\partial\tilde{\boldsymbol{\mu}}^{H}(\boldsymbol{\varsigma})}{\partial\boldsymbol{\varsigma}[i]}\boldsymbol{C}_{\tilde{\boldsymbol{v}}}^{-1}\frac{\partial\tilde{\boldsymbol{\mu}}(\boldsymbol{\varsigma})}{\partial\boldsymbol{\varsigma}[j]}\Big),\forall i,j\in\{1,\ldots,T\},
\end{equation}
}\vspace{-0.2cm}
\textup{where $\frac{\partial}{\partial}(\cdot)$ denotes the partial
derivative operator.}
\end{lem}
Notice that based on \eqref{eq:vector form}, we have $\tilde{\boldsymbol{y}}\sim\mathcal{CN}\big(\mathrm{vec}(\boldsymbol{A}_{r}^{c}\boldsymbol{B}\boldsymbol{A}_{t}^{H}\boldsymbol{X}),\sigma_{r}^{2}\boldsymbol{I}\big)$.
Let $\boldsymbol{F}_{\boldsymbol{\xi}}\in\mathbb{R}^{3M\times3M}$
denote the FIM for estimating $\boldsymbol{\xi}$. Based on Lemma
1, the element in the $i$-th row and $j$-th column of $\boldsymbol{F}_{\boldsymbol{\xi}}$,
$\forall i,j\in\{1,\ldots,3M\}$, is given by\vspace{-0.2cm}
\begin{eqnarray}
\boldsymbol{F}_{\boldsymbol{\xi}}[i,j] & = & \frac{2}{\sigma_{r}^{2}}\mathrm{Re}\Big(\frac{\partial\mathrm{vec}(\boldsymbol{A}_{r}^{c}\boldsymbol{B}\boldsymbol{A}_{t}^{H}\boldsymbol{X})^{H}}{\partial\boldsymbol{\boldsymbol{\xi}}[i]}\cdot\frac{\partial\mathrm{vec}(\boldsymbol{A}_{r}^{c}\boldsymbol{B}\boldsymbol{A}_{t}^{H}\boldsymbol{X})}{\partial\boldsymbol{\boldsymbol{\xi}}[j]}\Big),
\end{eqnarray}
\vspace{-0.2cm}
where $\dot{\boldsymbol{A}_{r}}=[\frac{\partial\boldsymbol{a}_{r}(\theta_{1})}{\partial\theta_{1}},\frac{\partial\boldsymbol{a}_{r}(\theta_{2})}{\partial\theta_{2}},\cdots,\frac{\partial\boldsymbol{a}_{r}(\theta_{M})}{\partial\theta_{M}}],$
and $\dot{\boldsymbol{A}_{t}}=[\frac{\partial\boldsymbol{a}_{t}(\theta_{1})}{\partial\theta_{1}},\frac{\partial\boldsymbol{a}_{t}(\theta_{2})}{\partial\theta_{2}},\cdots,\frac{\partial\boldsymbol{a}_{t}(\theta_{M})}{\partial\theta_{M}}].$
Following the similar derivation procedures in \cite{li2007range},
we have the FIM in a block matrix form:\vspace{-0.2cm}
\begin{equation}
\boldsymbol{F}_{\boldsymbol{\xi}}=\frac{2}{\sigma_{r}^{2}}\left[\begin{array}{ccc}
\mathrm{Re}(\boldsymbol{F}_{11}) & \mathrm{Re}(\boldsymbol{F}_{12}) & -\mathrm{Im}(\boldsymbol{F}_{12})\\
\mathrm{Re}(\boldsymbol{F}_{12})^{T} & \mathrm{Re}(\boldsymbol{F}_{22}) & -\mathrm{Im}(\boldsymbol{F}_{22})\\
-\mathrm{Im}(\boldsymbol{F}_{12})^{T} & -\mathrm{Im}(\boldsymbol{F}_{22})^{T} & \mathrm{Re}(\boldsymbol{F}_{22})
\end{array}\right],\label{eq:FIM}
\end{equation}
\vspace{-0.2cm}
where\vspace{-0.3cm}
\begin{eqnarray}
\boldsymbol{F}_{11} & = & L(\dot{\boldsymbol{A}_{r}}^{T}\dot{\boldsymbol{A}_{r}}^{c})\odot(\boldsymbol{B}^{c}\boldsymbol{A}_{t}^{T}\boldsymbol{S}_{x}^{T}\boldsymbol{A}_{t}^{c}\boldsymbol{B}^{T})+L(\dot{\boldsymbol{A}_{r}}^{T}\boldsymbol{A}_{r}^{c})\odot(\boldsymbol{B}^{c}\boldsymbol{A}_{t}^{T}\boldsymbol{S}_{x}^{T}\dot{\boldsymbol{A}_{t}}^{c}\boldsymbol{B}^{T})\nonumber \\
 &  & +L(\boldsymbol{A}_{r}^{T}\dot{\boldsymbol{A}_{r}}^{c})\odot(\boldsymbol{B}^{c}\dot{\boldsymbol{A}_{t}}^{T}\boldsymbol{S}_{x}^{T}\boldsymbol{A}_{t}^{c}\boldsymbol{B}^{T})+L(\boldsymbol{A}_{r}^{T}\boldsymbol{A}_{r}^{c})\odot(\boldsymbol{B}^{c}\dot{\boldsymbol{A}_{t}}^{T}\boldsymbol{S}_{x}^{T}\dot{\boldsymbol{A}_{t}}^{c}\boldsymbol{B}^{T}),\\
\boldsymbol{F}_{12} & = & L(\dot{\boldsymbol{A}_{r}}^{T}\boldsymbol{A}_{r}^{c})\odot(\boldsymbol{B}^{c}\boldsymbol{A}_{t}^{T}\boldsymbol{S}_{x}^{T}\boldsymbol{A}_{t}^{c})+L(\boldsymbol{A}_{r}^{T}\boldsymbol{A}_{r}^{c})\odot(\boldsymbol{B}^{c}\dot{\boldsymbol{A}_{t}}^{T}\boldsymbol{S}_{x}^{T}\boldsymbol{A}_{t}^{c}),\\
\boldsymbol{F}_{22} & = & L(\boldsymbol{A}_{r}^{T}\boldsymbol{A}_{r}^{c})\odot(\boldsymbol{A}_{t}^{T}\boldsymbol{S}_{x}^{T}\boldsymbol{A}_{t}^{c}).
\end{eqnarray}
\vspace{-0.2cm}
Then, the corresponding CRB matrix for estimating $\boldsymbol{\xi}$
is\vspace{-0.2cm}
\begin{equation}
\boldsymbol{C}_{2}=\boldsymbol{F}_{\boldsymbol{\xi}}^{-1}.\label{eq:CRB2}
\end{equation}
\vspace{-0.2cm}
Similar as in \eqref{eq:CRB1}, we adopt the trace of the CRB matrix
$\boldsymbol{C}_{2}$ as the scalar performance metric for estimating
$\boldsymbol{\xi}$, i.e.,\vspace{-0.2cm}
\begin{equation}
\mathrm{CRB}_{2}(\boldsymbol{S}_{x})=\mathrm{tr}(\boldsymbol{C}_{2})=\mathrm{tr}(\boldsymbol{F}_{\boldsymbol{\xi}}^{-1}).
\end{equation}
\vspace{-0.2cm}

\subsection{Problem Formulation}

We are interested in revealing the fundamental C-R performance tradeoff
limits of the multicast multi-target ISAC system for both Scenarios
I and II. Notice that in our considered setup, the transmit multi-beam
signals $\boldsymbol{X}$ are reused for both sensing and communications.
This ensures the optimal performance, as adopting separate communication
and sensing signal beams would lead to the degradation in both communication
and sensing performances, due to the potential mutual interference
between communication and sensing signals. Let $\mathcal{C}_{1}$
and $\mathcal{C}_{2}$ denote the achievable C-R regions of the ISAC
system in Scenario I and II, respectively, which are defined as\vspace{-0.2cm}
\begin{align}
\mathcal{C}_{i}\overset{\triangle}{=}\underset{\boldsymbol{S}_{x}\succeq\boldsymbol{0}}{\bigcup}\big\{(\hat{R},\hat{\Psi})|\hat{R}\leq R(\boldsymbol{S}_{x}), & \hat{\Psi}\geq\mathrm{CRB}_{i}(\boldsymbol{S}_{x}),\mathrm{tr}(\boldsymbol{S}_{x})\leq P\big\},i\in\{1,2\}.
\end{align}
Our objective is to characterize the Pareto boundary of each region
$\mathcal{C}_{i},i\in\{1,2\}$, at each point of which we cannot improve
one criteria without sacrificing the other \cite{Haocheng2022,xiong2022flowing}.
Note that the multicast rate $R(\boldsymbol{S}_{x})$ in \eqref{eq:rate},
$\mathrm{CRB}_{1}(\boldsymbol{S}_{x})$ in \eqref{eq:CRB}, and $\mathrm{CRB}_{2}(\boldsymbol{S}_{x})$
in \eqref{eq:CRB2} are both convex with respect to (w.r.t.) $\boldsymbol{S}_{x}$.
Therefore, the C-R regions $\mathcal{C}_{i},i\in\{1,2\}$, are both
convex sets. As a result, we can characterize the whole Pareto boundary
by finding each boundary point via minimizing the CRB subject to varying
multicast rate thresholds \cite{Haocheng2022,xiong2022flowing,boyd2004convex}.

In particular, we optimize the transmit covariance matrix $\boldsymbol{S}_{x}$
to minimize the estimation CRB for multi-target estimation while ensuring
a minimum multicast communication rate $\bar{R}$, subject to a maximum
transmit power constraint in \eqref{eq:Average power}. First, consider
Scenario I, for which the multicast-rate-constrained CRB minimization
problem is formulated as\vspace{-0.2cm}
\begin{eqnarray}
(\mathrm{P1}): & \underset{\boldsymbol{S}_{x}\succeq\boldsymbol{0}}{\min} & \frac{N_{r}\sigma_{r}^{2}}{L}\mathrm{tr}(\boldsymbol{S}_{x}^{-1})\nonumber \\
 & \mathrm{s.t.} & \underset{k\in\mathcal{K}}{\min}\Big\{\log_{2}\Big(1+\frac{\boldsymbol{h}_{k}^{H}\boldsymbol{S}_{x}\boldsymbol{h}_{k}}{\sigma^{2}}\Big)\Big\}\geq\bar{R},\nonumber \\
 &  & \mathrm{tr}(\boldsymbol{S}_{x})\leq P.\label{eq:CRB minimization problem1}
\end{eqnarray}
\vspace{-0.2cm}
By introducing $\Gamma=\sigma^{2}(2^{\bar{R}}-1)$, problem (P1) is
equivalently reformulated as\vspace{-0.2cm}
 
\begin{subequations}
\begin{eqnarray}
(\mathrm{P1.1}): & \underset{\boldsymbol{S}_{x}\succeq\boldsymbol{0}}{\min} & \mathrm{tr}(\boldsymbol{S}_{x}^{-1})\nonumber \\
 & \mathrm{s.t.} & \boldsymbol{h}_{k}^{H}\boldsymbol{S}_{x}\boldsymbol{h}_{k}\geq\Gamma,\forall k\in\mathcal{K},\label{eq:rate constraint}\\
 &  & \mathrm{tr}(\boldsymbol{S}_{x})\leq P.\label{eq:Power constraint}
\end{eqnarray}
\end{subequations}
 In problem (P1.1), the objective is a trace inverse function that
is strictly convex w.r.t. $\boldsymbol{S}_{x}$ \cite{horn2012matrix},
and the constraints in \eqref{eq:rate constraint} and \eqref{eq:Power constraint}
are both affine. As a result, problem (P1.1) is a convex optimization
problem. In Section III, we will obtain its globally optimal solution
in a semi-closed form by applying the Lagrange duality method \cite{boyd2004convex}.

Next, we consider Scenario II, for which the multicast-rate-constrained
CRB minimization problem is formulated as\vspace{-0.2cm}
\begin{eqnarray}
(\mathrm{P2}): & \underset{\boldsymbol{S}_{x}\succeq\boldsymbol{0}}{\min} & \mathrm{CRB}_{2}(\boldsymbol{S}_{x})\nonumber \\
 & \mathrm{s.t.} & \underset{k\in\mathcal{K}}{\min}\Big\{\log_{2}\Big(1+\frac{\boldsymbol{h}_{k}^{H}\boldsymbol{S}_{x}\boldsymbol{h}_{k}}{\sigma^{2}}\Big)\Big\}\geq\bar{R},\nonumber \\
 &  & \mathrm{tr}(\boldsymbol{S}_{x})\leq P.\label{eq:CRB minimization scenario II}
\end{eqnarray}
\vspace{-0.2cm}
Then, problem (P2) is equivalently reformulated as\vspace{-0.2cm}
\begin{subequations}
\begin{eqnarray}
\mathrm{(P2.1):} & \underset{\boldsymbol{S}_{x}\succeq\boldsymbol{0}}{\min} & \mathrm{CRB}_{2}(\boldsymbol{S}_{x})\nonumber \\
 & \mathrm{s.t.} & \boldsymbol{h}_{k}^{H}\boldsymbol{S}_{x}\boldsymbol{h}_{k}\geq\Gamma,\forall k\in\mathcal{K},\label{eq:rate constraint 2}\\
 &  & \mathrm{tr}(\boldsymbol{S}_{x})\leq P.\label{eq:power constraint2}
\end{eqnarray}
\end{subequations}
In Section IV, we will find the optimal solution to problem (P2.1)
or equivalently (P2) by using the semi-definite programing (SDP) technique.

It is worth noting that by solving the multicast-rate-constrained
CRB minimization problem (P1) or (P2) under one given rate threshold
$\bar{R}$, we can generally find one Pareto boundary point at C-R
region $\mathcal{C}_{1}$ for Scenario I or $\mathcal{C}_{2}$ for
Scenario II. By adjusting the value of $\bar{R}$, we can obtain the
complete Pareto boundary points on C-R region $\mathcal{C}_{i}$,
$i\in\{1,2\}$, for achieving different C-R tradeoffs \cite{Haocheng2022,xiong2022flowing}.
\begin{rem}
\label{remark1} Before proceeding to solve problems (P1) and (P2)
for ISAC, we discuss the two special cases with sole information multicasting
and sole multi-target sensing, respectively. First, considering the
information multicasting only, the maximum multicast capacity can
be found by solving the following problem:\vspace{-0.2cm}
\begin{eqnarray}
 & \underset{\boldsymbol{S}_{x}\succeq\boldsymbol{0}}{\max} & \underset{k\in\mathcal{K}}{\min}\Big\{\log_{2}\Big(1+\frac{\boldsymbol{h}_{k}^{H}\boldsymbol{S}_{x}\boldsymbol{h}_{k}}{\sigma^{2}}\Big)\Big\}\nonumber \\
 & \mathrm{s.t.} & \mathrm{tr}(\boldsymbol{S}_{x})\leq P.\label{eq:capacity problem}
\end{eqnarray}
It has been shown in \cite{jindal2006capacity} that problem \eqref{eq:capacity problem}
is optimally solvable via SDP. Let $\boldsymbol{S}_{x}^{\mathrm{com}}$
denote the optimal solution to problem \eqref{eq:capacity problem}.
Accordingly, we have the maximum multicast capacity as $R_{\mathrm{max}}=R(\boldsymbol{S}_{x}^{\mathrm{com}})$,
and the corresponding CRB in Scenario $i\in\{1,2\}$ as $\mathrm{CRB}_{i}(\boldsymbol{S}_{x}^{\mathrm{com}})$.
As such, we obtain a corner Pareto boundary point of the C-R region
$\mathcal{C}_{i}$ as $(R_{\mathrm{max}},\mathrm{CRB}_{i}(\boldsymbol{S}_{x}^{\mathrm{com}}))$,
which corresponds to the maximum multicast rate. Note that in Scenario
I, if $\boldsymbol{S}_{x}^{\mathrm{com}}$ is rank deficient, i.e.,
not a full-rank matrix, we have $\mathrm{CRB}_{1}(\boldsymbol{S}_{x}^{\mathrm{com}})\rightarrow\infty$,
which means the transmit degrees of freedom (DoFs) are not sufficient
to estimate the complete multi-target response matrix $\boldsymbol{G}$.
Next, we consider the multi-target sensing only, in which the CRB
minimization problem is expressed as follows, where $i=1$ and $i=2$
correspond to Scenarios I and II, respectively.\vspace{-0.2cm}
\begin{eqnarray}
 & \underset{\boldsymbol{S}_{x}\succeq\boldsymbol{0}}{\min} & \mathrm{CRB}_{i}(\boldsymbol{S}_{x})\nonumber \\
 & \mathrm{s.t.} & \mathrm{tr}(\boldsymbol{S}_{x})\leq P.\label{eq:minimum crb}
\end{eqnarray}
For Scenario I, it has been shown in \cite{liu2021cramer} that the
optimal solution to problem \eqref{eq:minimum crb} is $\boldsymbol{S}_{x}^{\mathrm{sen},\mathrm{1}}=\frac{P}{N_{t}}\boldsymbol{I}$,
i.e., the isotropic transmission is optimal. For Scenario II, it has
been proved in \cite{li2007range} that the optimal solution $\boldsymbol{S}_{x}^{\mathrm{sen},\mathrm{2}}$
to problem \eqref{eq:minimum crb} can be obtained via \ac{sdp}.
With $\boldsymbol{S}_{x}^{\mathrm{sen},i}$ at hand, the minimum CRB
is obtained as $\mathrm{CRB}_{\min}^{i}=\mathrm{CRB}_{i}(\boldsymbol{S}_{x}^{\mathrm{sen},i})$,
and the corresponding multicast rate becomes $R(\boldsymbol{S}_{x}^{\mathrm{sen},i})$.
We thus obtain another corner Pareto boundary point $\big(R(\boldsymbol{S}_{x}^{\mathrm{sen},i}),\mathrm{CRB}_{\min}^{i}\big)$
at the C-R region $\mathcal{C}_{i}$, which corresponds to multi-target
CRB minimization.\vspace{-0.2cm}
\end{rem}

\section{Optimal Solution to Problem (P1.1) or (P1) for Scenario I}

In this section, we present the optimal solution to the multicast-rate-constrained
CRB minimization problem (P1.1) for Scenario I. Notice that problem
(P1.1) is convex and satisfies the Slater's condition \cite{boyd2004convex}.
As a result, the strong duality holds between problem (P1.1) and its
dual problem. Therefore, we adopt the Lagrange duality method to find
the optimal solution to (P1.1) and analyze its structure to gain insights.
The optimal solution to (P1.1) or equivalently (P1) is obtained in
the following proposition. 
\begin{prop}
\textup{\label{thm:1}Suppose that }$\lambda^{\mathrm{opt}}$ \textup{and}
$\{\mu_{k}^{\mathrm{opt}}\}$\textup{ denote the optimal dual solutions
to the dual problem of (P1.1), which are associated to constraints
\eqref{eq:Power constraint} and \eqref{eq:rate constraint}, respectivaly.
Define $\boldsymbol{A}(\lambda^{\mathrm{opt}},\{\mu_{k}^{\mathrm{opt}}\})\overset{\triangle}{=}\lambda^{\mathrm{opt}}\boldsymbol{I}-\sum_{k=1}^{K}\mu_{k}^{\mathrm{opt}}\boldsymbol{h}_{k}\boldsymbol{h}_{k}^{H}\succeq\boldsymbol{0}$,
for which the \ac{evd} is $\boldsymbol{A}(\lambda^{\mathrm{opt}},\{\mu_{k}^{\mathrm{opt}}\})=\boldsymbol{U}^{\mathrm{opt}}\boldsymbol{\Lambda}^{\mathrm{opt}}\boldsymbol{U}^{\mathrm{opt}H}$,
where $\boldsymbol{\Lambda}^{\mathrm{opt}}=\mathrm{diag}(\alpha_{1}^{\mathrm{opt}},\dots,\alpha_{N_{t}}^{\mathrm{opt}})$
. Then it must follow that $\mathrm{rank}(\boldsymbol{A}(\lambda^{\mathrm{opt}},\{\mu_{k}^{\mathrm{opt}}\}))=N_{t}$
and accordingly $\alpha_{1}^{\mathrm{opt}}\ge\dots\ge\alpha_{N_{t}}^{\mathrm{opt}}>0$.
The optimal solution $\boldsymbol{S}_{x}^{\mathrm{opt}}$ to problem
(P1.1) is given by} \vspace{-0.3cm}
\begin{align}
\boldsymbol{S}_{x}^{\mathrm{opt}}= & \boldsymbol{U}^{\mathrm{opt}}{\boldsymbol{\Sigma}^{\mathrm{opt}}}\boldsymbol{U}^{\mathrm{opt}H},
\end{align}
\textup{where ${\boldsymbol{\Sigma}^{\mathrm{opt}}}=\mathrm{diag}(\tau_{1}^{\mathrm{opt}},\dots,\tau_{N_{t}}^{\mathrm{opt}}$)
and $\tau_{i}^{\mathrm{opt}}=(\alpha_{i}^{\mathrm{opt}})^{-1/2}$,
$\forall i\in\{1,\dots,N_{t}\}$.}\vspace{-0.2cm}
\end{prop}
\begin{proof}
Please refer to Appendix A.
\end{proof}
Based on Proposition 1, we have the following remark to reveal more
insights on the optimal solution to (P1.1) or (P1).
\begin{rem}
With the optimal dual solution $\lambda^{\mathrm{opt}}$ and $\{\mu_{k}^{\mathrm{opt}}\}$,
we define the weighted communication channel of the $K$ CUs as $\boldsymbol{H}=\sum_{k=1}^{K}\mu_{k}^{\mathrm{opt}}\boldsymbol{h}_{k}\boldsymbol{h}_{k}^{H}\succeq\boldsymbol{0}$,
with rank $N_{\text{com}}=\text{rank}(\boldsymbol{H})\le N_{t}$.
The EVD of $\boldsymbol{H}$ is then expressed as $\boldsymbol{H}=[{\boldsymbol{U}}_{\text{sen}}~{\boldsymbol{U}}_{\text{com}}]{\boldsymbol{\Delta}}[{\boldsymbol{U}}_{\text{sen}}~{\boldsymbol{U}}_{\text{com}}]^{H}$,
where ${\boldsymbol{\Delta}}=\mathrm{diag}(0,\ldots,0,\delta_{1},\ldots,\delta_{N_{\text{com}}})$
with $\delta_{1}\ge\ldots\ge\delta_{N_{\text{com}}}\geq0$ denoting
the $N_{\text{com}}$ eigenvalues, and $\boldsymbol{U}_{\text{com}}\in\mathbb{C}^{N_{t}\times N_{\text{com}}}$
and ${\boldsymbol{U}}_{\text{sen}}\in\mathbb{C}^{N_{t}\times(N_{t}-N_{\text{com}})}$
consist of the eigenvectors corresponding to the non-zero and zero
eigenvalues, thus spanning the communication subspaces and sensing
subspaces, respectively. In this case, recall that $\boldsymbol{A}(\lambda^{\mathrm{opt}},\{\mu_{k}^{\mathrm{opt}}\})=\lambda^{\mathrm{opt}}\boldsymbol{I}-\sum_{k=1}^{K}\mu_{k}^{\mathrm{opt}}\boldsymbol{h}_{k}\boldsymbol{h}_{k}^{H}=\lambda^{\mathrm{opt}}\boldsymbol{I}-\boldsymbol{H}\succ\boldsymbol{0}$.
It is evident that the EVD of $\boldsymbol{A}(\lambda^{\mathrm{opt}},\{\mu_{k}^{\mathrm{opt}}\})$
can be expressed as 
\begin{align}
\boldsymbol{A}(\lambda^{\mathrm{opt}},\{\mu_{k}^{\mathrm{opt}}\})= & \boldsymbol{U}^{\mathrm{opt}}\boldsymbol{\Lambda}^{\mathrm{opt}}\boldsymbol{U}^{\mathrm{opt}H}=[{\boldsymbol{U}}_{\text{sen}}~{\boldsymbol{U}}_{\text{com}}]\big(\lambda^{\mathrm{opt}}\boldsymbol{I}-{\boldsymbol{\Delta}}\big)[{\boldsymbol{U}}_{\text{sen}}~{\boldsymbol{U}}_{\text{com}}]^{H},
\end{align}
where $\boldsymbol{U}^{\mathrm{opt}}=[{\boldsymbol{U}}_{\text{sen}}~\boldsymbol{U}_{\text{com}}]$
and $\boldsymbol{\Lambda}^{\mathrm{opt}}=\mathrm{diag}(\alpha_{1}^{\mathrm{opt}},\dots,\alpha_{N_{t}}^{\mathrm{opt}})$
with\vspace{-0.2cm}
\begin{equation}
\alpha_{i}^{\mathrm{opt}}=\left\{ \begin{array}{cc}
\lambda^{\mathrm{opt}}, & {\text{if~}}i=1,\ldots,N_{t}-N_{\text{com}}\\
\lambda^{\mathrm{opt}}-\delta_{i-(N_{t}-N_{\text{com}})}, & {\text{if~}}i=N_{t}-N_{\text{com}}+1,\ldots,N_{t}
\end{array}\right..
\end{equation}
\vspace{-0.2cm}
In this case, it follows from Proposition \ref{thm:1} that the optimal
transmit covariance is given by\vspace{-0.2cm}
\begin{align}
\boldsymbol{S}_{x}^{\mathrm{opt}} & =\boldsymbol{U}{\boldsymbol{\Sigma}^{\mathrm{opt}}}\boldsymbol{U}^{\mathrm{opt}H}={\boldsymbol{U}}_{\text{sen}}{\boldsymbol{\Sigma}}_{\text{sen}}^{\mathrm{opt}}{\boldsymbol{U}}_{\text{sen}}^{H}+\boldsymbol{U}_{\text{com}}{\boldsymbol{\Sigma}}_{\text{com}}^{\mathrm{opt}}\boldsymbol{U}_{\text{com}}^{H},
\end{align}
where ${\boldsymbol{\Sigma}}_{\text{sen}}^{\mathrm{opt}}=\mathrm{diag}(\sqrt{1/\lambda^{\mathrm{opt}}},\ldots,\sqrt{1/\lambda^{\mathrm{opt}}})$
and ${\boldsymbol{\Sigma}}_{\text{com}}^{\mathrm{opt}}=\mathrm{diag}(\sqrt{1/(\lambda^{\mathrm{opt}}-\delta_{1})},\ldots,\sqrt{1/(\lambda^{\mathrm{opt}}-\delta_{N_{\text{com}}})})$.
It is clear that the transmit covariance can be decomposed into two
parts, including $\boldsymbol{U}_{\text{com}}{\boldsymbol{\Sigma}}_{\text{com}}^{\mathrm{opt}}\boldsymbol{U}_{\text{com}}^{H}$
towards the CUs for ISAC, and ${\boldsymbol{U}}_{\text{sen}}{\boldsymbol{\Sigma}}_{\text{sen}}^{\mathrm{opt}}{\boldsymbol{U}}_{\text{sen}}^{H}$
for dedicated sensing. In the first signal part for dedicated sensing,
equal power allocation is adopted. In contrast, in the second signal
part for ISAC, more transmit power is allocated over the sub-channels
with stronger combined channel gains ($\alpha_{i}^{\mathrm{opt}}$).\vspace{-0.2cm}
\end{rem}

\section{Optimal Solution to Problem (P2.1) or (P2) for Scenario II}

In this section, we present the optimal solution to problem (P2.1)
or (P2) by exploiting the SDP technique. Towards this end, we re-express
$\boldsymbol{F}_{\boldsymbol{\xi}}^{-1}[i,i]$ as $\boldsymbol{e}_{i}^{T}\boldsymbol{F}_{\boldsymbol{\xi}}^{-1}\boldsymbol{e}_{i}$,
$\forall i\in\{1,...,3M\}$, where $\boldsymbol{e}_{i}$ denotes the
$i$-th column vector of an identity matrix with proper dimension.
By introducing a set of auxiliary variables $\{t_{i}\}_{i=1}^{3M}$,
problem (P2.1) is equivalently reformulated as\vspace{-0.2cm}
 
\begin{subequations}
\begin{eqnarray}
\mathrm{(P2.2):} & \underset{\boldsymbol{S}_{x}\succeq\boldsymbol{0},\{t_{i}\}_{i=1}^{3M}}{\min} & \sum_{i=1}^{3M}t_{i}\nonumber \\
 & \mathrm{s.t.} & t_{i}\geq\boldsymbol{e}_{i}^{T}\boldsymbol{F}_{\boldsymbol{\xi}}^{-1}\boldsymbol{e}_{i},\forall i\in\{1,\dots,3M\},\label{eq:Auxiliary variables}\\
 &  & \textrm{\eqref{eq:rate constraint 2} and \eqref{eq:power constraint2}}.\nonumber 
\end{eqnarray}
\end{subequations}
\vspace{-0.2cm}
Next, by exploiting the Schur component, the constraints in \eqref{eq:Auxiliary variables}
are equivalently re-expressed as\vspace{-0.2cm}
\begin{equation}
\left[\begin{array}{cc}
\boldsymbol{F}_{\boldsymbol{\xi}} & \boldsymbol{e}_{i}\\
\boldsymbol{e}_{i}^{T} & t_{i}
\end{array}\right]\succeq\boldsymbol{0},\forall i\in\{1,\dots,3M\}.\label{eq:Auxiliary variables 2}
\end{equation}
\vspace{-0.2cm}
By replacing \eqref{eq:Auxiliary variables} in (P2.2) as \eqref{eq:Auxiliary variables 2},
problem (P2.2) is further equivalently reformulated into 
\begin{subequations}
\begin{eqnarray}
(\mathrm{P2.3}): & \underset{\boldsymbol{S}_{x}\succeq\boldsymbol{0},\{t_{i}\}_{i=1}^{3M}}{\min} & \sum_{i=1}^{3M}t_{i}\nonumber \\
 & \mathrm{s.t.} & \textrm{\textrm{\eqref{eq:rate constraint 2}, \eqref{eq:power constraint2}}, and \eqref{eq:Auxiliary variables 2}}.
\end{eqnarray}
\end{subequations}
\vspace{-0.1cm}
 Notice that the constraints in \eqref{eq:Auxiliary variables 2}
are linear matrix inequality (LMI) constraints w.r.t. $\boldsymbol{S}_{x}$
and the constraints in \eqref{eq:rate constraint 2} and \eqref{eq:power constraint2}
are affine. As a result, (P2.3) is a standard convex SDP, which can
be efficiently solved via convex optimization tools, such as CVX \cite{cvx}.
Therefore, we obtain the optimal solution of $\boldsymbol{S}_{x}$
to problem (P2.3), and thus problem (P2), denoted by $\boldsymbol{S}_{x}^{\mathrm{opt,2}}$.
\begin{rem}
It is worth discussing the structure of the optimal solution $\boldsymbol{S}_{x}^{\mathrm{opt,2}}$
for gaining insights. First, it is observed from \eqref{eq:CRB2}
that only the transmit signal components lying in the subspaces spanned
by $\left[\boldsymbol{A}_{t},\dot{\boldsymbol{A}_{t}}\right]$ contribute
to $\mathrm{CRB}_{2}(\boldsymbol{S}_{x})$. Similarly, it is observed
from \eqref{eq:rate} that only the signal components lying in the
subspaces spanned by $\left[\boldsymbol{h}_{1},\dots,\boldsymbol{h}_{K}\right]$
contribute to multicast communication rate $R(\boldsymbol{S}_{x})$.
In this case, suppose that the rank of the accumulated matrix $\left[\boldsymbol{A}_{t},\dot{\boldsymbol{A}_{t}},\boldsymbol{h}_{1},\dots,\boldsymbol{h}_{K}\right]\in\mathbb{C}^{N_{t}\times(2M+K)}$
is $J$, i.e., $\mathrm{rank}\big(\left[\boldsymbol{A}_{t},\dot{\boldsymbol{A}_{t}},\boldsymbol{h}_{1},\dots,\boldsymbol{h}_{K}\right]\big)=J\le N_{t}$,
and its truncated singular value decomposition (SVD) is\vspace{-0.2cm}
\begin{equation}
\left[\boldsymbol{A}_{t},\dot{\boldsymbol{A}_{t}},\boldsymbol{h}_{1},\dots,\boldsymbol{h}_{K}\right]=\boldsymbol{U}\boldsymbol{\Lambda}\boldsymbol{V}^{H},
\end{equation}
where $\boldsymbol{U}\in\mathbb{C}^{N_{t}\times J}$ and $\boldsymbol{V}\in\mathbb{C}^{(2M+K)\times J}$
collect the left and right singular vectors corresponding to the non-zero
singular values, respectively. In this case, we can express the transmit
covariance matrix $\boldsymbol{S}_{x}$ as $\boldsymbol{S}_{x}={\boldsymbol{U}}^{H}\bar{\boldsymbol{S}_{x}}{\boldsymbol{U}}$,
where $\bar{\boldsymbol{S}_{x}}\in\mathbb{C}^{J\times J}$ corresponds
to the equivalent transmit covariance matrix to be decided.\textcolor{blue}{{}
}This intuitively shows that the transmit covariance matrix $\boldsymbol{S}_{x}$
should lie in the subspaces spanned by the sensing and communication
channels for efficient ISAC. In the case when $M$ and $K$ (and thus
$J$) are small but $N_{t}$ is large, we can optimize $\bar{\boldsymbol{S}_{x}}$
instead of $\boldsymbol{S}_{x}$, in order to significantly reduce
the dimension of optimization variables and thus reduce the solution
complexity.
\end{rem}

\section{Joint Information and Sensing Beamforming Design }

Sections III and IV presented the optimal transmit covariance solutions
to the multicast-rate-constrained CRB minimization problems (P1) and
(P2) for Scenarios I and II, respectively, which are of high rank
in general, especially when the number of CUs $K$ and the number
of targets $M$ become large. Note that this high-rank covariance
is typically achieved through multi-stream transmission. While the
transmitted streams may interfere with each other, it is remarkable
that separate codeword decoding can still achieve capacity by utilizing
a linear minimum mean-square error (LMMSE) successive interference
cancellation (SIC) receiver. The basic idea behind the LMMSE-SIC receiver
is as follows: for each target stream, we convey the observed sequence
to a corresponding LMMSE filter and then decode that particular stream.
Subsequently, we reconstruct the interference caused by that stream
and cancel it before decoding the next stream. Note that although
this LMMSE-SIC approach is capacity achieving, it requires high decoding
complexity\cite{heath2018foundations}. To resolve this issue, this
section presents  joint transmit information and sensing beamforming
designs, in which the BS adopts one single information beam for delivering
the common message, together with multiple dedicated sensing beams\footnote{Notice that for problems (P1) and (P2) introducing the dedicated sensing
beams is not necessary. This is due to the fact that the transmit
covariance $\boldsymbol{S}_{x}$ in (P1) and (P2) is of general rank,
and thus adding dedicated sensing beams does not improve the ISAC
performance. }.\vspace{-0.2cm}

\subsection{Joint Transmit Beamforming }

In the joint transmit beamforming, the BS uses one information beam
to deliver the common message. Let $\boldsymbol{w}\in\mathbb{C}^{N_{t}\times1}$
denote the information beamforming vector and $s_{\text{com}}(n)$
denote the common message carried by symbol $n$ with $s_{\text{com}}(n)\sim\mathcal{CN}(0,1)$.
In addition, the BS also employs dedicated sensing beams for providing
additional DoFs to facilitate the multi-target sensing. Let $\boldsymbol{x}_{\mathrm{sen}}(n)$
denote the dedicated sensing signals to provide additional DoFs for
estimating the targets parameters, which is a pseudo random vector
with mean zero and covariance matrix $\mathbb{E}(\boldsymbol{x}_{\mathrm{sen}}(n)\boldsymbol{x}_{\mathrm{sen}}^{H}(n))=\boldsymbol{S}_{\mathrm{sen}}\succeq\boldsymbol{0}$.
Without loss of generality, we assume that $\mathrm{rank}(\boldsymbol{S}_{\mathrm{sen}})=r_{\mathrm{sen}}\ge0$,
which corresponds to a number of $r_{\mathrm{sen}}$ sensing beams
that can be obtained by performing the EVD on $\boldsymbol{S}_{\mathrm{sen}}$
\cite{hua2021optimal}. In this case, the transmit signal is given
by\vspace{-0.2cm}
\begin{equation}
\boldsymbol{x}(n)=\boldsymbol{w}s_{\text{com}}(n)+\boldsymbol{x}_{\mathrm{sen}}(n),
\end{equation}
for which the corresponding transmit covariance matrix is given by\vspace{-0.2cm}
\begin{equation}
\boldsymbol{S}_{x}=\boldsymbol{S}_{\mathrm{sen}}+\boldsymbol{w}\boldsymbol{w}^{H}.
\end{equation}
The transmit power constraint in \eqref{eq:Average power} thus becomes\vspace{-0.2cm}
\begin{equation}
\mathrm{tr}(\boldsymbol{S}_{x})=\mathrm{tr}(\boldsymbol{S}_{\mathrm{sen}})+\|\boldsymbol{w}\|^{2}\le P.
\end{equation}

First, consider the multicast communication. The received signal at
CU $k\in\mathcal{K}$ is given as\vspace{-0.2cm}
\begin{equation}
\tilde{y_{k}}(n)=\boldsymbol{h}_{k}^{H}\boldsymbol{w}s_{\text{com}}(n)+\boldsymbol{h}_{k}^{H}\boldsymbol{x}_{\mathrm{sen}}(n)+z_{k}(n).\label{eq:BF signal}
\end{equation}
In general, the dedicated sensing signals $\boldsymbol{h}_{k}^{H}\boldsymbol{x}_{\mathrm{sen}}(n)$
may introduce harmful interference at the receiver of the CUs. We
impose an assumption that the CU possesses \textit{a priori} knowledge
of the radar sensing sequence and is aware of the CSI such that the
interference caused by the radar signal can be efficiently eliminated
from $\tilde{y_{k}}(n)$ in \eqref{eq:BF signal}. As a result, we
consider that the CU is capable of such interference cancellation
before information decoding \cite{hua2021optimal}.\footnote{The case without sensing interference cancellation will be considered
as a benchmark for performance evaluation in Section VI.} In this case, the SNR at the receiver of CU $k\in\mathcal{K}$ is
denoted as\vspace{-0.2cm}
\begin{equation}
\hat{\gamma_{k}}=\frac{\boldsymbol{h}_{k}^{H}\boldsymbol{w}\boldsymbol{w}^{H}\boldsymbol{h}_{k}}{\sigma^{2}},
\end{equation}
and the corresponding achievable multicast rate is\vspace{-0.2cm}
\begin{equation}
\hat{R}(\boldsymbol{w})\overset{\triangle}{=}\underset{k\in\mathcal{K}}{\min}~\Big\{\log_{2}\Big(1+\frac{\boldsymbol{h}_{k}^{H}\boldsymbol{w}\boldsymbol{w}^{H}\boldsymbol{h}_{k}}{\sigma^{2}}\Big)\Big\}.
\end{equation}

Next, we consider the multi-target sensing, in which both information
and sensing beams are jointly employed for target estimation. In this
case, the estimation CRB for Scenarios I and II are expressed as $\mathrm{CRB}_{1}(\boldsymbol{S}_{\mathrm{sen}}+\boldsymbol{w}\boldsymbol{w}^{H})$
and $\mathrm{CRB}_{2}(\boldsymbol{S}_{\mathrm{sen}}+\boldsymbol{w}\boldsymbol{w}^{H})$,
respectively.

Based on the derived multicast rate and estimation CRB above, the
multicast-rate-constrained CRB minimization problems via joint information
and sensing beamforming are formulated as problems (P3) and (P4) for
Scenarios I and II, respectively:\vspace{-0.2cm}
\begin{eqnarray}
(\mathrm{P3}): & \underset{\boldsymbol{w},\boldsymbol{S}_{\mathrm{sen}}\succeq\boldsymbol{0}}{\min} & \mathrm{CRB}_{1}(\boldsymbol{S}_{\mathrm{sen}}+\boldsymbol{w}\boldsymbol{w}^{H})\nonumber \\
 & \mathrm{s.t.} & \underset{k\in\mathcal{K}}{\min}~\Big\{\log_{2}\Big(1+\frac{\boldsymbol{h}_{k}^{H}\boldsymbol{w}\boldsymbol{w}^{H}\boldsymbol{h}_{k}}{\sigma^{2}}\Big)\Big\}\geq\bar{R},\label{eq:rate3}\\
 &  & \mathrm{tr}(\boldsymbol{S}_{\mathrm{sen}}+\boldsymbol{w}\boldsymbol{w}^{H})\leq P.\label{eq:BF1}
\end{eqnarray}
\begin{eqnarray}
(\mathrm{P4}): & \underset{\boldsymbol{w},\boldsymbol{S}_{\mathrm{sen}}\succeq\boldsymbol{0}}{\min} & \mathrm{CRB}_{2}(\boldsymbol{S}_{\mathrm{sen}}+\boldsymbol{w}\boldsymbol{w}^{H})\nonumber \\
 & \mathrm{s.t.} & \underset{k\in\mathcal{K}}{\min}~\Big\{\log_{2}\Big(1+\frac{\boldsymbol{h}_{k}^{H}\boldsymbol{w}\boldsymbol{w}^{H}\boldsymbol{h}_{k}}{\sigma^{2}}\Big)\Big\}\geq\bar{R},\\
 &  & \mathrm{tr}(\boldsymbol{S}_{\mathrm{sen}}+\boldsymbol{w}\boldsymbol{w}^{H})\leq P.\label{eq:BF2}
\end{eqnarray}

\subsection{SCA-Based Solution to Joint Beamforming Problems (P3) and (P4)}

This subsection addresses problems (P3) and (P4). Notice that the
two problems are both non-convex and thus are generally difficult
to be optimally solved. To tackle this issue, we propose to apply
the SCA technique to find high-quality solutions. As the SCA algorithms
for solving problems (P3) and (P4) are similar, we only focus on solving
problem (P3) in the following. The similar algorithm can be adopted
to solve problem (P4), for which the details are omitted for brevity.

Consider problem (P3). First, by substituting $\boldsymbol{S}_{\mathrm{sen}}=\boldsymbol{S}_{x}-\boldsymbol{w}\boldsymbol{w}^{H}$
and introducing $\Gamma=\sigma^{2}(2^{\bar{R}}-1)$, we equivalently
reformulate problem (P3) as\vspace{-0.2cm}
 
\begin{subequations}
\begin{align}
(\mathrm{P3.1}):\underset{\boldsymbol{w},\boldsymbol{S}_{x}}{\min}\  & \ensuremath{\mathrm{tr}}(\ensuremath{\boldsymbol{S}_{x}^{-1}})\nonumber \\
\mathrm{s.t.} & \ \boldsymbol{h}_{k}^{H}\boldsymbol{w}\boldsymbol{w}^{H}\boldsymbol{h}_{k}\geq\Gamma,\forall k\in\mathcal{K},\label{eq:SINR constraintCancel}\\
 & \mathrm{tr}(\boldsymbol{S}_{x})\leq P,\label{eq:22-2}\\
 & \boldsymbol{S}_{x}-\boldsymbol{w}\boldsymbol{w}^{H}\succeq\boldsymbol{0}.\label{eq:23-2}
\end{align}
\end{subequations}
 Based on the Schur component, problem (P3.1) is equivalently reformulated
as \vspace{-0.2cm}
\begin{subequations}
\begin{align}
(\mathrm{P3.2}):\underset{\boldsymbol{w},\boldsymbol{S}_{x}}{\min}\  & \ensuremath{\mathrm{tr}}(\ensuremath{\boldsymbol{S}_{x}^{-1}})\nonumber \\
\mathrm{s.t.}~ & \textrm{\eqref{eq:SINR constraintCancel} and \eqref{eq:22-2}},\nonumber \\
 & \left[\begin{array}{cc}
\boldsymbol{S}_{x} & \boldsymbol{w}\\
\boldsymbol{w}^{H} & 1
\end{array}\right]\succeq\boldsymbol{0},\label{eq:23-2-1}
\end{align}
\end{subequations}
 where constraint \eqref{eq:23-2} in (P3.1) is replaced by the LMI
constraint in \eqref{eq:23-2-1}. However, the reformulated problem
(P3.2) is still non-convex due to the non-convex quadratic constraints.

Next, we propose a SCA-based algorithm to find a high-quality solution
to (P3.2) in an iterative manner, in which problem (P3.2) is approximated
as a series of convex problems. Specifically, consider a particular
iteration $i$, with $\boldsymbol{w}^{(i)}$ and $\boldsymbol{S}_{x}^{(i)}$
denoting the local solution point. Based on the first-order approximation,
it follows that \vspace{-0.3cm}
\begin{equation}
\boldsymbol{h}_{k}^{H}\boldsymbol{w}\boldsymbol{w}^{H}\boldsymbol{h}_{k}\geq2\mathrm{Re}(\boldsymbol{w}^{H}\boldsymbol{h}_{k}\boldsymbol{h}_{k}^{H}\boldsymbol{w}^{(i)})-\boldsymbol{w}^{(i)H}\boldsymbol{h}_{k}\boldsymbol{h}_{k}^{H}\boldsymbol{w}^{(i)}.
\end{equation}
By replacing $\boldsymbol{h}_{k}^{H}\boldsymbol{w}\boldsymbol{w}^{H}\boldsymbol{h}_{k}$
as $\mathrm{Re}(\boldsymbol{w}^{H}\boldsymbol{h}_{k}\boldsymbol{h}_{k}^{H}\boldsymbol{w}^{(i)})-\boldsymbol{w}^{(i)H}\boldsymbol{h}_{k}\boldsymbol{h}_{k}^{H}\boldsymbol{w}^{(i)}$,
we approximate the non-convex quadratic constraints in \eqref{eq:SINR constraintCancel}
as \vspace{-0.3cm}
\begin{eqnarray}
 & 2\mathrm{Re}(\boldsymbol{w}^{H}\boldsymbol{h}_{k}\boldsymbol{h}_{k}^{H}\boldsymbol{w}^{(i)})-\boldsymbol{w}^{(i)H}\boldsymbol{h}_{k}\boldsymbol{h}_{k}^{H}\boldsymbol{w}^{(i)}\geq\Gamma,\forall k\in\mathcal{K}.\label{eq:first order expansion-1}
\end{eqnarray}
Accordingly, we obtain a restricted convex version of problem (P3.2)
in iteration $i$ as \vspace{-0.3cm}
\begin{eqnarray*}
(\mathrm{P3.3\textrm{-}}i): & \underset{\boldsymbol{w},\boldsymbol{S}_{x}}{\max} & \ensuremath{\mathrm{tr}}(\ensuremath{\boldsymbol{S}_{x}^{-1}})\\
 & \mathrm{s.t.} & \textrm{\eqref{eq:first order expansion-1}, \eqref{eq:22-2}, and \eqref{eq:23-2-1}.}
\end{eqnarray*}
\vspace{-0.2cm}
Notice that problem ($\textrm{P3.3-}i$) is convex and thus can be
optimally solved via numerical tools such as CVX. Let $\hat{\boldsymbol{w}}_{\mathrm{opt}}^{(i)}$
denote the optimal solution to problem ($\textrm{P3.3-}i$) in iteration
$i$, which is updated as the local point for the next iteration,
i.e., $\boldsymbol{w}^{(i+1)}=\hat{\boldsymbol{w}}_{\mathrm{opt}}^{(i)}$.
Note that in order to implement the SCA based algorithm, we need to
find a feasible initial point. In particular, we first find the transmit
information beamforming to maximize the multicast rate by using the
algorithm in \cite{sidiropoulos2006transmit} based on the \ac{sdr}
and Gaussian randomization. Suppose that $\hat{\boldsymbol{w}}$ denote
the obtained transmit information beamformer. Accordingly, we use
$\boldsymbol{w}^{(0)}=\hat{\boldsymbol{w}}$ as the initial point
for the SCA-based algorithms.

\subsection{Convergence and Computational Complexity }

In this subsection, we provide the convergence analysis for the proposed
SCA based algorithm. Due to the lower-bound approximation in \eqref{eq:first order expansion-1},
it can be shown that the optimal solution $\hat{\boldsymbol{w}}_{\mathrm{opt}}^{(i)}$
in each iteration $i$ corresponds to a feasible point for next iteration
the next iteration $i+1$. As such, the SCA iterations are ensured
to result in monotonically decreasing CRB objective values. Therefore,
the convergence of the SCA-based algorithm for handling problem (P3.2)
and thus (P3) is guaranteed \cite{razaviyayn2014successive}. Also,
following \cite{razaviyayn2014successive}, the iteration can converge
to a Karush-Kuhn-Tucker (KKT) point. 

Next, we analyze the computational complexity. First, with a maximum
iteration number $N_{\max}$, the iteration complexity is given as
$\mathcal{O}(N_{\max})$, where $\mathcal{O}(\cdot)$ is the big-O
operator. Next, we analyze the worst-case computational complexity
of solving problem (P3.2). Problem (P3.2) is a convex problem that
can be effectively solved by CVX via the interior-point method. Problem
(P3.2) involves $(N+1)$ affine constraints and one LMI constraints
with dimension $(N+1)$. Based on \cite{WangAntJ14}, the order of
computational complexity is $\mathcal{O}\big(\sqrt{N+1}N^{3}\ln(\frac{1}{\varepsilon})\big)$
where $\varepsilon>0$ is the duality gap accuracy. By combining them,
we finally have the total computational complexity as $\mathcal{O}\Big(N_{\max}N^{3.5}\ln(\frac{1}{\varepsilon})\Big)$.\vspace{-0.1cm}

\section{Numerical Results\vspace{-0.2cm}
}

This section provides numerical results to validate the performance
of our proposed designs. In the following, we first show the C-R tradeoff
performances achieved by our proposed designs as compared to other
benchmark schemes and then show the practical target estimation performances
by considering practical  estimators. For convenience, we normalize
the noise powers as $\sigma^{2}=\sigma_{r}^{2}=0\mathrm{\textrm{ dBm}}$. 

\subsection{C-R Tradeoff Performance }

This subsection presents the C-R tradeoff performance achieved by
our proposed optimal transmit covariance solution and the joint beamforming
design, as compared to the following two benchmark schemes. 
\begin{itemize}
\item \textbf{Isotropic transmission}: The BS adopts the isotropic transmission
by setting the transmit covariance as $\boldsymbol{S}_{x}=\frac{P}{N_{t}}\boldsymbol{I}=\boldsymbol{S}_{x}^{\mathrm{sen},\mathrm{1}}$,
which is optimal in minimizing the estimation CRB in Scenario I (see
Remark \ref{remark1}). This scheme is only feasible when the achieved
multicast rate $R(\boldsymbol{S}_{x}^{\mathrm{sen},\mathrm{1}})$
is no smaller than the rate threshold $\bar{R}$. \vspace{-0.1cm}
\item \textbf{Beamforming without sensing interference cancellation}: The
BS sends one information beam $\boldsymbol{w}s_{\text{com}}(n)$ together
with dedicated sensing beams $\boldsymbol{x}_{\mathrm{sen}}(n)$.
Different from the design in Section IV, we consider that the CUs
are not capable of canceling the interference caused by sensing signals
$\boldsymbol{h}_{k}^{H}\boldsymbol{x}_{\mathrm{sen}}(n)$. As a result,
the received SINR becomes $\frac{\boldsymbol{h}_{k}^{H}\boldsymbol{w}\boldsymbol{w}^{H}\boldsymbol{h}_{k}}{\boldsymbol{h}_{k}^{H}\boldsymbol{S}_{s}\boldsymbol{h}_{k}+\sigma^{2}}$.
The joint beamforming designs in this case correspond to problems
(P3) and (P4) by replacing \eqref{eq:rate3} as \vspace{-0.1cm}
\begin{equation}
\underset{k\in\mathcal{K}}{\min}\Big\{\log_{2}\Big(1+\frac{\boldsymbol{h}_{k}^{H}\boldsymbol{w}\boldsymbol{w}^{H}\boldsymbol{h}_{k}}{\boldsymbol{h}_{k}^{H}\boldsymbol{S}_{s}\boldsymbol{h}_{k}+\sigma^{2}}\Big)\Big\}\geq\bar{R}.\label{eq:SINR2}
\end{equation}
Although the resultant beamforming optimization problems are non-convex,
they can be solved similarly as for (P3) and (P4) by using the SCA
technique. We skip the details for brevity and the interested readers
can refer to \cite{ren2022fundamental} for the detailed algorithm
for solving the problem for Scenario I. 
\end{itemize}
\vspace{-0.1cm}

First, we consider Scenario I, in which Rayleigh fading is adopted
\cite{liu2021cramer} for modeling the communication channels from
the BS to the CUs. $\boldsymbol{h}_{k}$ is a CSCG random vector with
zero mean and identity covariance matrix. The total transmit power
budget at the BS is set as $P=10\textrm{ dBm}$ unless otherwise stated.
Fig. \ref{fig:1} shows the achieved C-R tradeoff performance, where
we set $N_{t}=N_{r}=10$ and $K=3$.\textcolor{blue}{{} }It is observed
that as the CRB increases, the achievable rate by the optimal covariance
solution approaches the asymptotical multicast capacity but can never
reach it. This is due to the fact that in this setup, the rate-maximization
transmit covariance $\boldsymbol{S}_{x}^{\mathrm{com}}$ (see Remark
\ref{remark1}) is rank-deficient such that $\mathrm{CRB_{com}}\rightarrow\infty$.
It is also observed that the optimal covariance solution significantly
outperforms the other three schemes and the proposed transmit beamforming
with sensing interference cancellation outperforms that without sensing
interference cancellation. This is because in this scenario, full-rank
transmission with multiple dedicated sensing beams is necessary for
effective target estimation, which may result in harmful sensing interference
making the sensing interference cancellation crucial.

\begin{figure}
\includegraphics[scale=0.5]{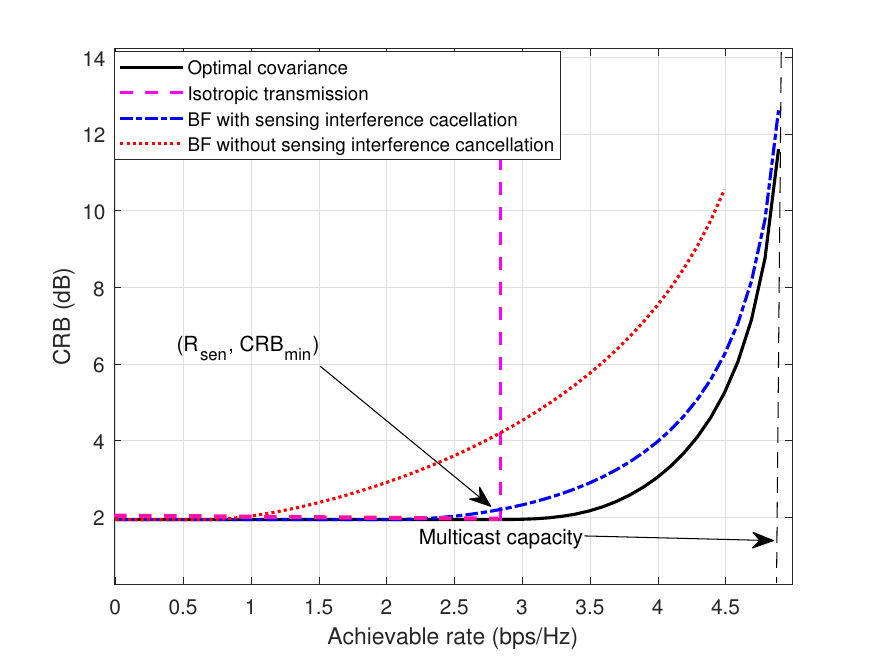}\centering\vspace{-0.3cm}
\caption{\label{fig:1}C-R tradeoff in Scenario I with $N_{t}=N_{r}=10$ and
$K=3$.}
\vspace{-0.8cm}
\end{figure}
\vspace{-0.1cm}

Figs. \ref{fig:2} and \ref{fig:3} show the achieved C-R tradeoff
performance with $N_{t}=N_{r}=4$, where $K=50$ and $K=500$, respectively.
In both figures, it is observed that when the CRB becomes sufficiently
large, the optimal covariance solution reaches the multicast capacity.
This is because the rate-maximization transmit covariance $\boldsymbol{S}_{x}^{\mathrm{com}}$
is of full-rank such that $\mathrm{CRB_{com}}$ is finite. Furthermore,
it is observed that the two transmit beamforming designs perform much
worse than the optimal covariance solution and the isotropic transmission
(w.r.t. the rate and CRB), while the performance gaps become more
substantial when $K$ becomes larger (i.e., $K=500$ in Fig. \ref{fig:3}
compared with $K=50$ in Fig. \ref{fig:2}). This is due to the fact
it is more likely that some CUs may have wireless channels orthogonal
to the information beam with the increasing number of CUs, thus making
the transmit beamforming design less efficient.

\begin{figure}
\centering%
\begin{minipage}[t]{0.48\textwidth}%
\includegraphics[scale=0.5]{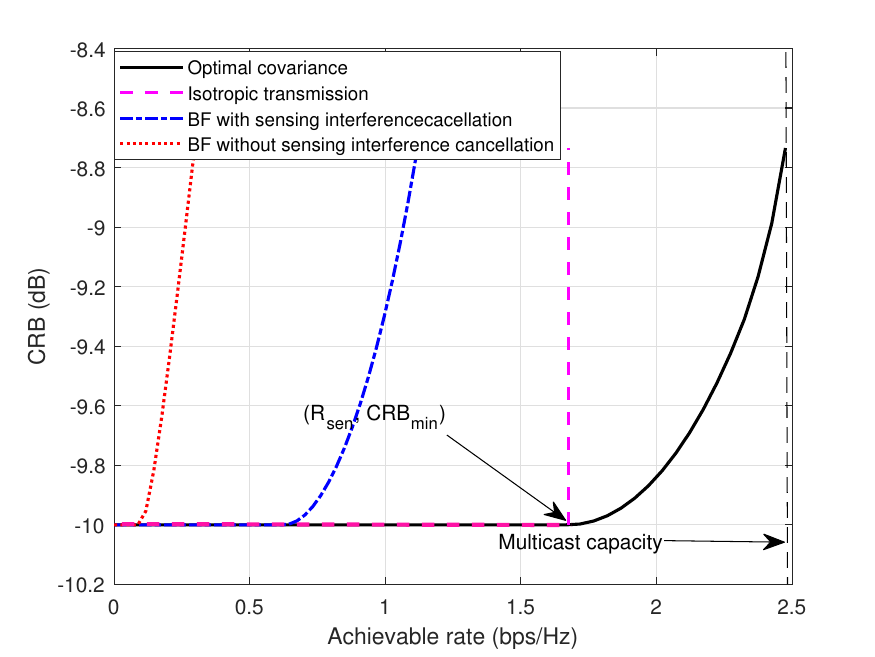}\centering\vspace{-0.3cm}
\caption{\label{fig:2}C-R tradeoff in Scenario I with $N_{t}=N_{r}=4$ and
$K=50$.}
\end{minipage}\hspace{0.15in}%
\begin{minipage}[t]{0.48\textwidth}%
 \centering\includegraphics[scale=0.5]{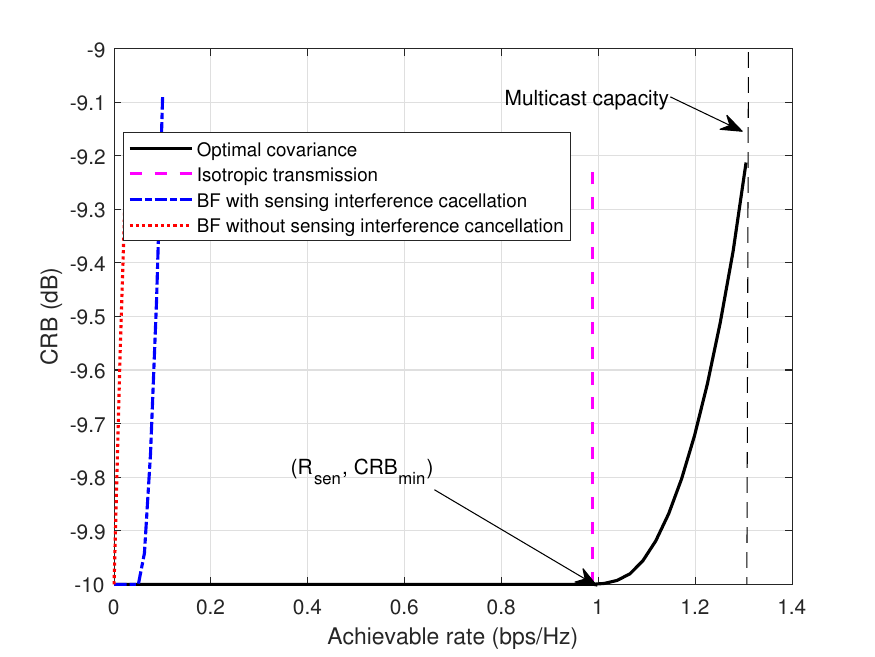}\vspace{-0.3cm}
\caption{\label{fig:3}C-R tradeoff in Scenario I with $N_{t}=N_{r}=4$ and
$K=500$.}
\end{minipage}\vspace{-0.8cm}
\end{figure}
\vspace{-0.1cm}

Fig. \ref{fig:4} shows the estimation CRB versus the transmit power
$P$, where $\bar{R}=1.2\textrm{ bps/Hz}$, $N_{t}=N_{r}=10$, and
$K=10$. It is observed that in the low transmit power regime, the
transmit beamforming with sensing interference cancellation performs
close to that without sensing interference cancellation, and inferior
to the optimal covariance solution. This is becasue in this case,
most power needs to be allocated for communication, thus making the
sensing interference cancellation less efficient. By contrast, in
the high transmit power regime, the transmit beamforming with sensing
interference cancellation is observed to perform close to the optimal
covariance solution and significantly outperform that without sensing
interference cancellation. This is due to the fact that in high SNR
regime, a single communication beam with equal power allocation is
able to satisfy the communication requirement. As a result, the sensing-optimal
isotropic transmission is adopted for both BF with sensing interference
cancellation and optimal design. This highlights the importance of
sensing interference cancellation in this case.

\begin{figure}
\includegraphics[scale=0.5]{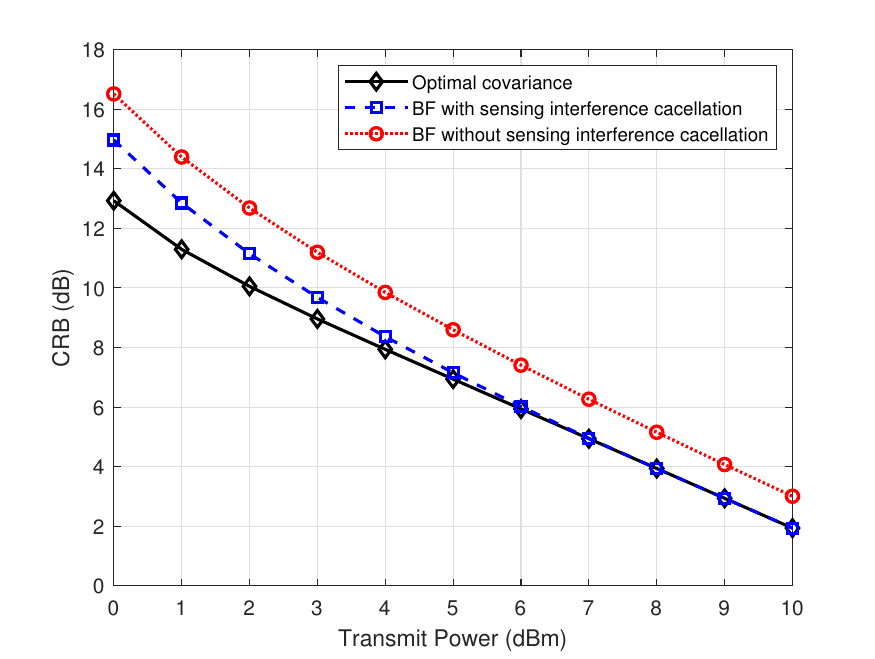}\centering\vspace{-0.3cm}
\caption{\label{fig:4}The CRB versus the transmit power $P$ in Scenario I,
where $\bar{R}=1.2\textrm{ bps/Hz}$, $N_{t}=N_{r}=10,$ and $K=10$.}
\vspace{-0.8cm}
\end{figure}

Next, we consider Scenario II. In this scenario, there are $M=2$
target located at $30{^\circ}$ and $-30{^\circ}$, with the reflection
coefficients being normalized to be $\beta_{1}=\beta_{2}=1$ \cite{xu2008target}.
We consider Rician fading for the wireless channel from the BS to
each CU, i.e., $\boldsymbol{h}_{k}=\sqrt{\frac{K_{r}}{K_{r}+1}}\boldsymbol{a}(\boldsymbol{\theta}_{K}[k])+\sqrt{\frac{1}{K_{r}+1}}\boldsymbol{h}_{k,\mathrm{NLoS}}$,
where $K_{r}=4\textrm{ dB}$ denotes the Rician factor, $\boldsymbol{\theta}_{K}[k]$
denotes the AoD for CU $k\in\mathcal{K}$, and $\boldsymbol{h}_{k,\mathrm{NLoS}}$
is a CSCG random vector with zero mean and identity covariance matrix.

Figs. \ref{fig:5} and \ref{fig:6} show the achieved C-R tradeoff
with $N_{t}=N_{r}=10,$ $P=10\textrm{ dBm}$, and $K=10$ CUs, where
the CU channels are correlated and uncorrelated with the sensing channels
(i.e., $\{\boldsymbol{\theta}_{K}[k]\}$ are randomly sampled in $\pm30{^\circ}\pm2{^\circ}$
and $\pm60{^\circ}\pm2{^\circ}$), respectively. It is observed in
Fig. \ref{fig:5} that the achieved CRB almost remains unchanged when
the rate threshold $\bar{R}$ is less than $3.5\textrm{ bps/Hz}$.
This is due to the fact that the rate constraint is inactive in this
regime, and thus we can minimize the CRB and achieve a satisfactory
communication rate concurrently by exploiting the correlated communication
and sensing channels. On the other hand, it is observed in Fig. \ref{fig:6}
that the achieved CRB monotonically increases as the required achievable
rate becomes larger. This is because when the communication and sensing
channels are less correlated, the communication and sensing signals
can hardly be reused for optimizing both performances, thus leading
to a more evident C-R tradeoff.

\begin{figure}
\centering%
\begin{minipage}[t]{0.48\textwidth}%
\includegraphics[scale=0.5]{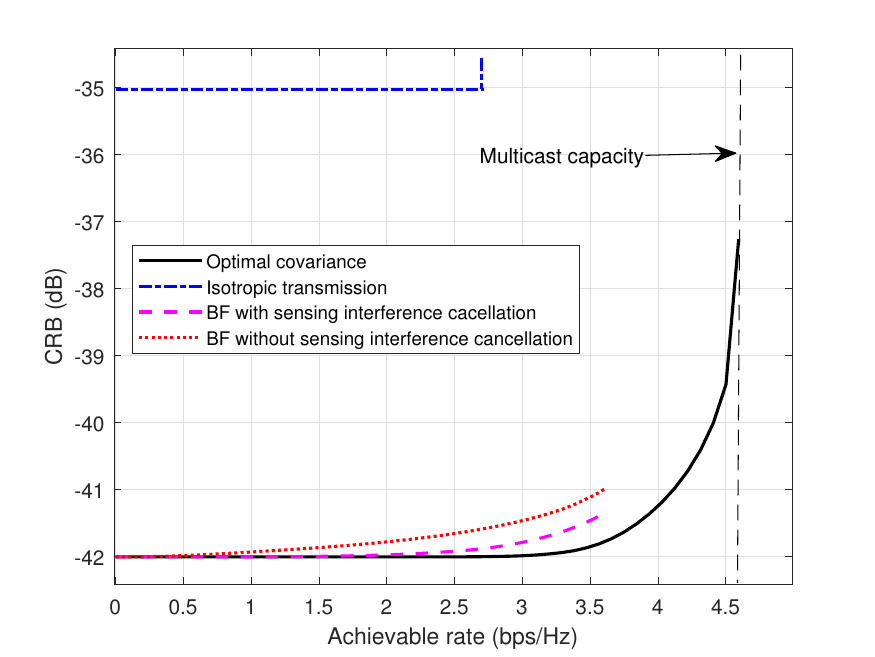}\vspace{-0.3cm}
\caption{\label{fig:5}C-R tradeoff in Scenario II with correlated communication
and sensing channels, where $K=10$, $N_{t}=N_{r}=10,$ and $P=10\textrm{ dBm}$.}
\end{minipage}\hspace{0.15in}%
\begin{minipage}[t]{0.48\textwidth}%
 \centering\includegraphics[scale=0.5]{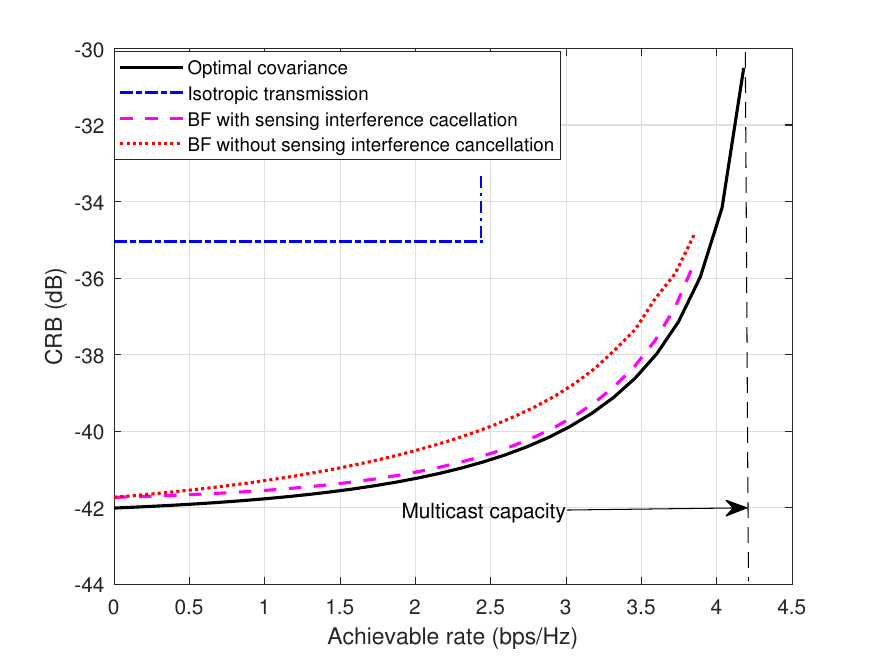}\vspace{-0.3cm}
\caption{\label{fig:6}C-R tradeoff in Scenario II with uncorrelated communication
and sensing channels, where $K=10$, $N_{t}=N_{r}=10,$ and $P=10\textrm{ dBm}$.}
\end{minipage}\vspace{-0.8cm}
\end{figure}
\vspace{-0.1cm}

Fig. \ref{fig:7} shows the CRB versus the transmit power ${P}$ for
Scenario II, with $\bar{R}=1.2\textrm{ bps/Hz}$, $N_{t}=N_{r}=10,$
and $K=10$, where the CUs' angles $\{\boldsymbol{\theta}_{K}[k]\}$
are randomly sampled between $-90{^\circ}$ and $90{^\circ}$. It
is observed that the proposed transmit beamforming with sensing interference
cancellation performs close to that without sensing interference cancellation
when $P$ is small and approaches that of the optimal covariance solution
when $P$ is large. This phenomenon is similar to Fig. \ref{fig:4}
for Scenario I, which can be similarly explained.

\begin{figure}
\includegraphics[scale=0.5]{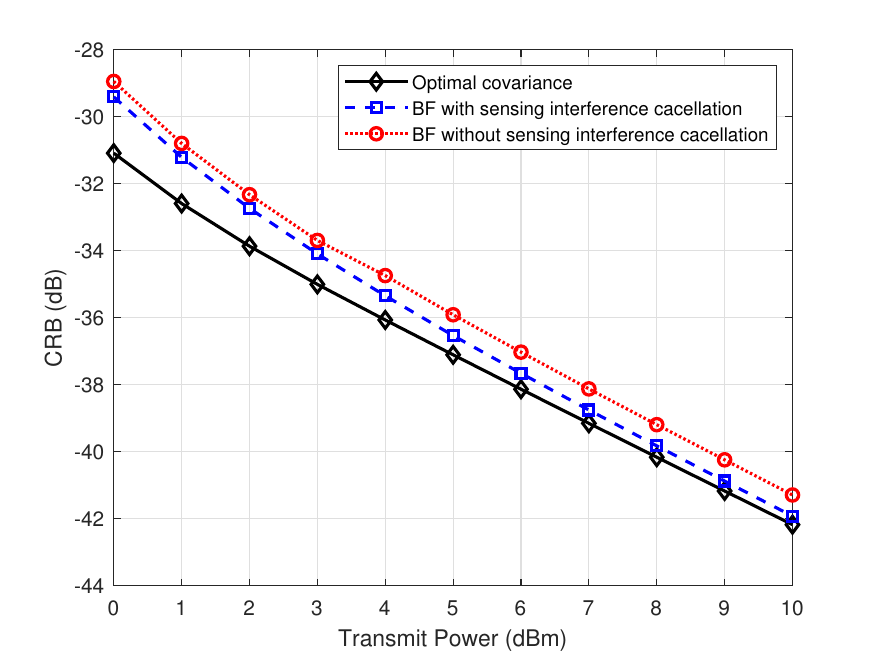}\centering\vspace{-0.3cm}
\caption{\label{fig:7}The CRB versus the transmit power $P$ in Scenario II,
where $\bar{R}=1.2\textrm{ bps/Hz}$, $N_{t}=N_{r}=10$, and $K=10$.}
\vspace{-0.8cm}
\end{figure}
\vspace{-0.2cm}

\subsection{Target Estimation Performance}

This subsection shows the exact estimation performance by considering
the practical estimators, i.e., least square (LS) estimator for Scenario
I and the Capon approximate maximum likelihood (CAML) estimator for
Scenario II \cite{liu2021cramer,xu2008target}. In the simulation,
we generate the transmitted signal $\boldsymbol{x}(n)$ based on the
given transmit covariance $\boldsymbol{S}_{x}$ as\vspace{-0.3cm}
\begin{equation}
\boldsymbol{x}(n)=\boldsymbol{V}\boldsymbol{u},
\end{equation}
where $\boldsymbol{S}_{x}=\boldsymbol{V}\boldsymbol{V}^{H}$ is the
Cholesky decomposition for $\boldsymbol{S}_{x}$ and $\boldsymbol{u}$
is a CSCG random vector with zero mean and identity covariance. In
this case, we obtain the received signal in \eqref{eq:Received signal for sensing}.
Accordingly, for Scenario I, the LS estimation of $\tilde{\boldsymbol{g}}$
based on \eqref{eq:Received signal for sensing} is obtained as \vspace{-0.3cm}
\begin{equation}
\tilde{\boldsymbol{g}}_{\mathrm{est}}=\big(\boldsymbol{X}^{c}\boldsymbol{X}^{T}\otimes\boldsymbol{I}_{N_{r}}\big)^{-1}(\boldsymbol{X}^{T}\otimes\boldsymbol{I}_{N_{r}})^{H}\tilde{\boldsymbol{y}}.
\end{equation}
In particular, we can find that $\tilde{\boldsymbol{g}}_{\mathrm{est}}\sim\mathcal{CN}\big(\tilde{\boldsymbol{g}},\sigma_{r}^{2}(\boldsymbol{X}^{c}\boldsymbol{X}^{T}\otimes\boldsymbol{I}_{N_{r}})^{-1}\big)$.
As a result, $\tilde{\boldsymbol{g}}_{\mathrm{est}}$ is a minimum
variance unbiased (MVU) estimator thus is CRB-achieving. For Scenario
II, we estimate $\boldsymbol{\beta}$ and $\boldsymbol{\theta}$ simultaneously
by adopting the CAML estimation \cite{xu2008target}, for which 20,000
angle grids in $\left[-90{^\circ},90{^\circ}\right]$ are considered
to ensure the resolution.

For comparison, we consider the beampattern-based design \cite{StoPETLiJ07}
as the benchmark scheme. The basic idea of beampattern-based design
is to allocate more power in desired target directions. Let $\boldsymbol{\theta}_{N}\in\mathbb{R}^{N\times1}$
denote the angle grid vector of interest, which is given as $[-\frac{\pi}{2}:\frac{\pi}{200}:\frac{\pi}{2}]$,
and $\boldsymbol{q}\in\mathbb{R}^{N\times1}$ denote the desired beampattern
in this grid. For Scenario I, we set\vspace{-0.3cm}
\begin{equation}
\boldsymbol{q}[n]=1,\forall n.
\end{equation}
For Scenario II, we set \vspace{-0.3cm}
\begin{equation}
\boldsymbol{q}[n]=\begin{cases}
1, & \exists m\in\mathcal{M},|\theta_{m}-\boldsymbol{\theta}_{N}[n]|\leq\Delta_{\theta},\\
0, & \textrm{others},
\end{cases}
\end{equation}
where $\Delta_{\theta}$ is a pre-defined angle width. Accordingly,
the beampattern-based design aims to minimize the beampattern matching
error while ensuring the multicast rate constraint at CUs, which is
formulated as\vspace{-0.3cm}
\begin{eqnarray}
 & \underset{\boldsymbol{S}_{x}\succeq\boldsymbol{0},\eta}{\min} & \sum_{i=1}^{N}|\eta\boldsymbol{q}[i]-\boldsymbol{a}^{H}(\boldsymbol{\theta}_{N}[i])\boldsymbol{S}_{x}\boldsymbol{a}^{H}(\boldsymbol{\theta}_{N}[i])|^{2}\nonumber \\
 & \mathrm{s.t.} & \underset{k\in\mathcal{K}}{\min}\Big\{\log_{2}\Big(1+\frac{\boldsymbol{h}_{k}^{H}\boldsymbol{S}_{x}\boldsymbol{h}_{k}}{\sigma^{2}}\Big)\Big\}\geq\bar{R},\nonumber \\
 &  & \mathrm{tr}(\boldsymbol{S}_{x})\leq P.\label{eq:BME}
\end{eqnarray}
Problem \eqref{eq:BME} can be transformed as a convex problem that
can be optimally solved via CVX.

In the simulation, we obtain the results by averaging over 1000 Monte
Carlo realizations. We suppose that there are $M=2$ targets located
at $-30{^\circ},30{^\circ}$. We set $N_{t}=N_{r}=10,K=10$, and $\bar{R}=0.5\textrm{ bps/Hz}$.

First, we evaluate the estimation performance in Scenario I. Fig.
\ref{fig:8} shows the estimation RMSE and the CRB versus the transmit
power $P$ with $L=64$, in which the actual RMSE and the theoretical
CRB using $\boldsymbol{S}_{x}$, and actually average CRB using $\frac{1}{L}\boldsymbol{XX}^{H}$
are shown. It is observed that our considered CRB-based design achieves
much better sensing performance than the beampattern-based benchmark.
Furthermore, the actual RMSE and actual CRB are observed to be almost
identical. This is due to the fact that the LS estimator is adopted
to estimate $\tilde{\boldsymbol{g}}$, which is generally a CRB-achieving
estimator.

Fig. \ref{fig:9} shows the RMSE and the CRB versus the transmit sequence
length $L$ for Scenario I, where $P=15\textrm{ dBm}$. It is observed
that the gap between the RMSE/CRB achieved by the sample covariance
$\frac{1}{L}\boldsymbol{XX}^{H}$ versus that by covariance $\boldsymbol{S}_{x}$
becomes negligible when the sequence length $L$ is greater than 64.
This means that the approximation in \eqref{eqn:app} is quite accurate
when $L\geq64$ in Scenario I.

\begin{figure}
\centering%
\begin{minipage}[t]{0.48\textwidth}%
\includegraphics[scale=0.5]{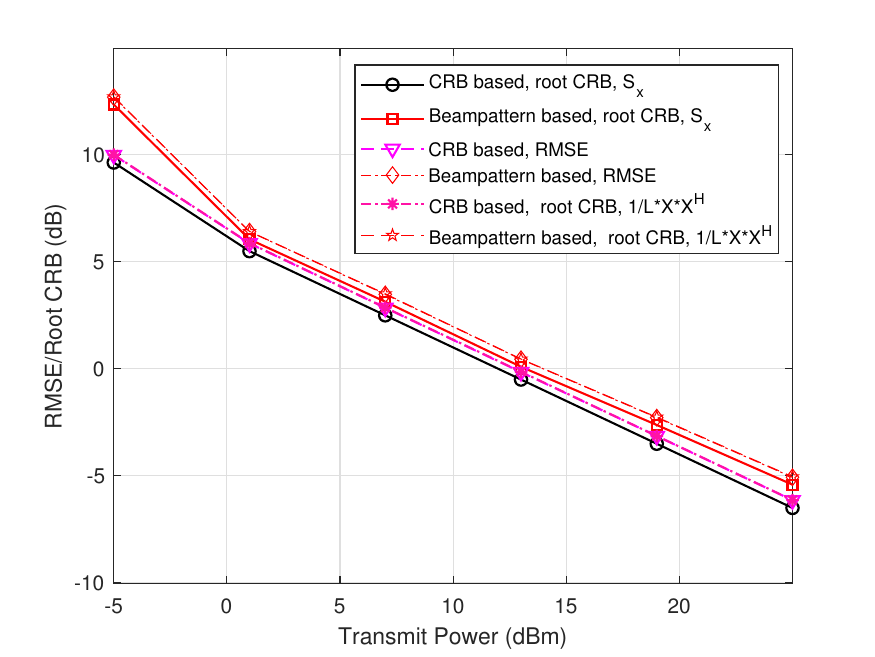}\centering\vspace{-0.3cm}
\caption{\label{fig:8}Estimation RMSE or root CRB versus transmit power $P$
for Scenario I, where $L=64$.}
\vspace{-0.1cm}
\end{minipage}\hspace{0.15in}%
\begin{minipage}[t]{0.48\textwidth}%
 \centering\includegraphics[scale=0.5]{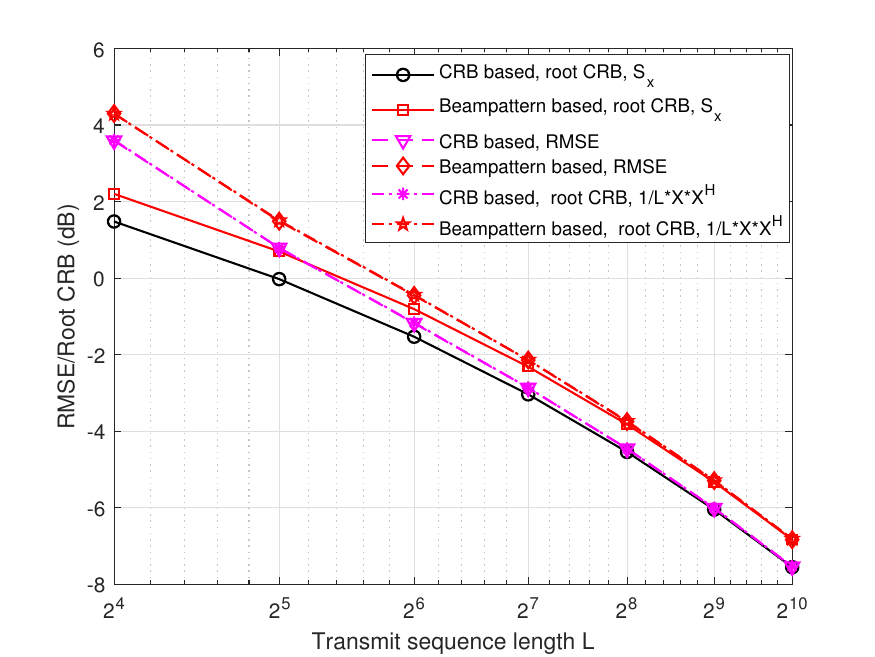}\vspace{-0.3cm}
\caption{\label{fig:9}Estimation RMSE or root CRB versus transmit sequence
length $L$ for Scenario I, where $P=15\textrm{ dBm}$.}
\vspace{-0.1cm}
\end{minipage}\vspace{-0.8cm}
\end{figure}

Next, we evaluate the estimation performance for Scenario II. Fig.
\ref{fig:10} shows the RMSE or the root CRB of angle estimation versus
the transmit power $P$ with $L=64$. It is observed that the CRB-based
design improves the practical estimation performance as compared to
the beampattern-based design. It is also observed that the actual
CRB using $\frac{1}{L}\boldsymbol{XX}^{H}$ is well approached by
the estimation RMSE. In this case, the gap between the RMSE and the
theoretical CRB using $\boldsymbol{S}_{x}$ mostly originates from
the approximation error.

Fig. \ref{fig:11} shows the estimation RMSE or root CRB of angle
estimation versus the transmit sequence length $L$ with $P=15\textrm{ dBm}$.
It is observed that for the CRB-based design, the achieved estimation
RMSE nearly achieves the CRB using $\boldsymbol{S}_{x}$ with $L\geq64$.
However, for the beampattern-based design, the achieved RMSE is unable
to achieve the corresponding CRB. This shows the benefits of CRB-based
design again.

Fig. \ref{fig:12} shows the RMSE or root CRB of amplitude estimation
with $L=64$. It is observed that there always exists a gap between
RMSE and CRB both in CRB- and beampattern-based designs. This gap
comes from the CAML estimator. It is also observed that the CRB-based
design actually improves the estimation RMSE around $4\textrm{ dB }$
compared to beampattern-based design.

Fig. \ref{fig:13} shows the estimation RMSE or root CRB of amplitude
estimation versus the transmit sequence length $L$ with $P=15\textrm{ dBm}$.
It is observed that there is almost no gap of CRB using $\frac{1}{L}\boldsymbol{XX}^{H}$
and $\boldsymbol{S}_{x}$ with $L\geq16$ while there is always around
$2~\textrm{dB}$ gap between the RMSE of amplitude estimation and
CRB, which is due to the CAML estimator.

\begin{figure}
\centering%
\begin{minipage}[t]{0.48\textwidth}%
\includegraphics[scale=0.5]{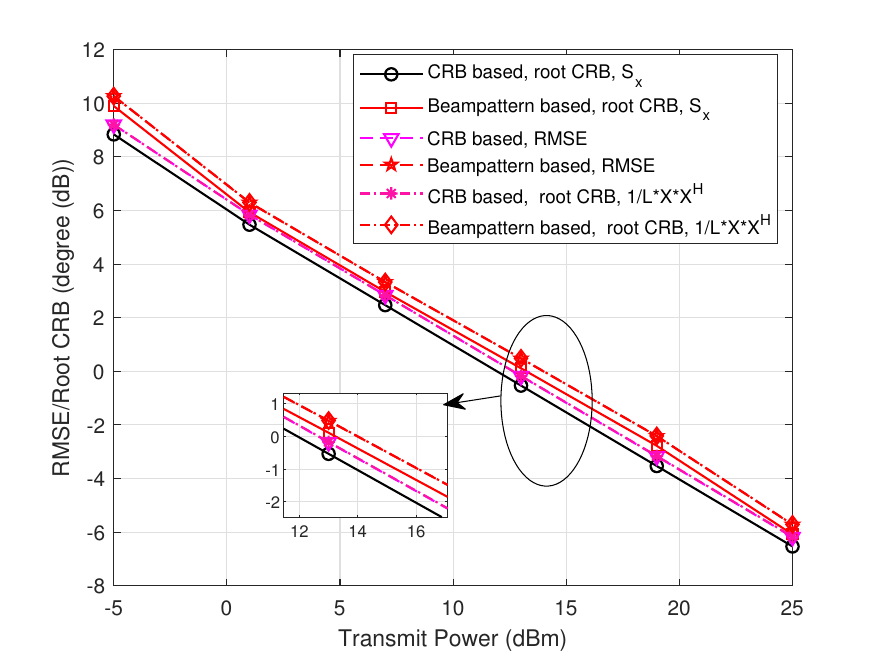}\centering\vspace{-0.3cm}
\caption{\label{fig:10}Estimation RMSE or root CRB of angle estimation versus
transmit power $P$ with $L=64$ in Senario II.}
\vspace{-0.1cm}
\end{minipage}\hspace{0.15in}%
\begin{minipage}[t]{0.48\textwidth}%
 \centering\includegraphics[scale=0.5]{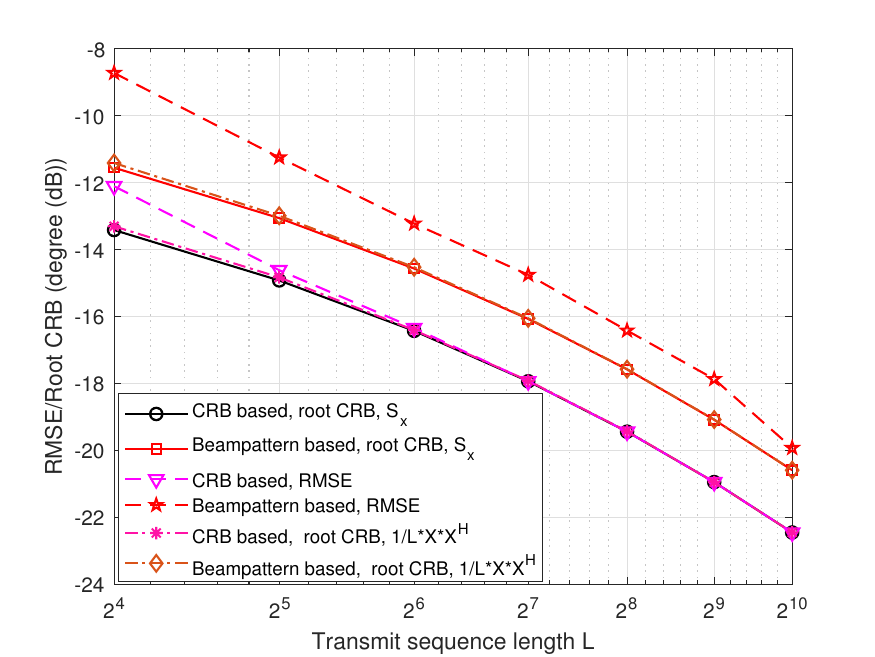}\vspace{-0.3cm}
\caption{\label{fig:11}Estimation RMSE or root CRB of angle estimation versus
transmit sequence length $L$ with $P=15\textrm{ dBm}$ in Senario
II.}
\vspace{-0.1cm}
\end{minipage}\vspace{-0.8cm}
\end{figure}

\begin{figure}
\centering%
\begin{minipage}[t]{0.48\textwidth}%
\includegraphics[scale=0.5]{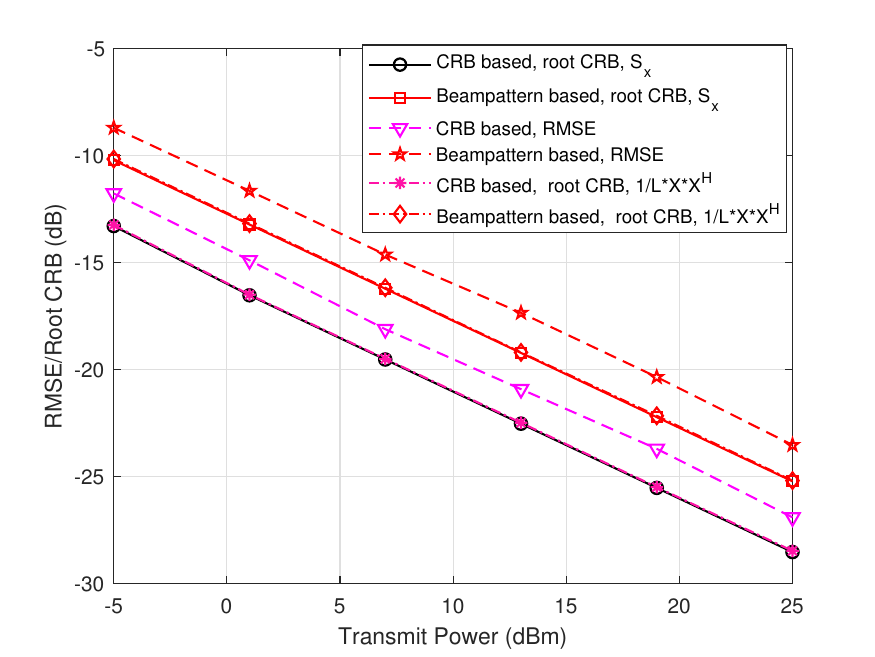}\centering\vspace{-0.3cm}
\caption{\label{fig:12}Estimation RMSE or root CRB of amplitude estimation
versus transmit power $P$ with $L=64$ in Senario II.}
\vspace{-0.1cm}
\end{minipage}\hspace{0.15in}%
\begin{minipage}[t]{0.48\textwidth}%
 \centering\includegraphics[scale=0.5]{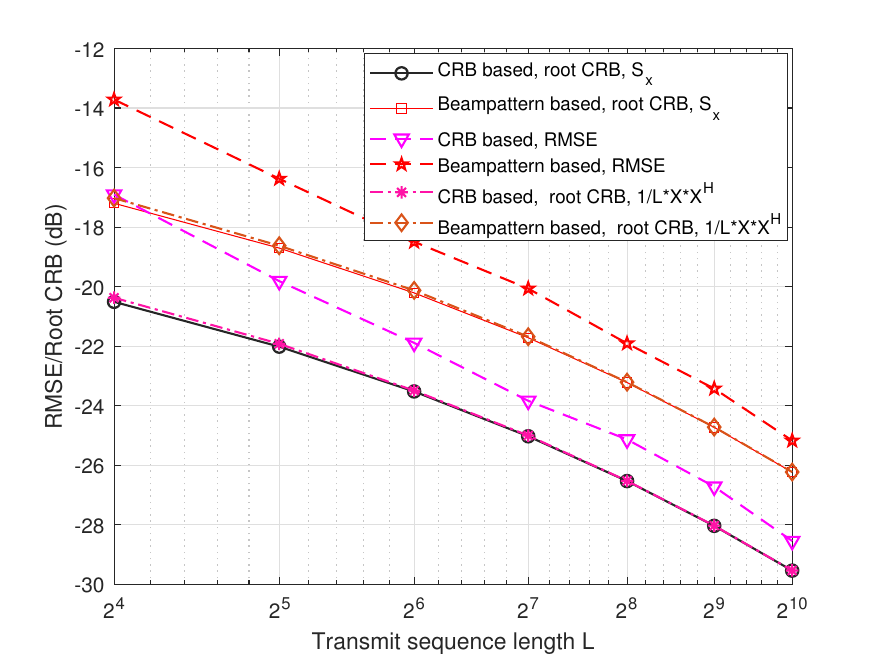}\vspace{-0.3cm}
\caption{\label{fig:13}Estimation RMSE or root CRB of amplitude estimation
versus transmit sequence length $L$ with $P=15\textrm{ dBm}$ in
Senario II.}
\vspace{-0.1cm}
\end{minipage}\vspace{-0.8cm}
\end{figure}

\section{Conclusion}

We studied the fundamental C-R tradeoff in multi-antenna multi-target
ISAC over multicast channels. In particular, we considered two scenarios,
namely Scenario I and Scenario II, in which the BS does not and does
have prior knowledge about targets, respectively. For both scenarios,
we characterized the fundamental C-R tradeoff by solving multicast-rate-constrained
CRB minimization problems via optimizing the transmit covariance matrix.
Next, we considered  joint information and sensing beamforming designs,
by exploiting dedicated sensing signal beams together with sensing
interference cancellation. Finally, we provided numerical results
to reveal the C-R tradeoff performance and practical estimation performance
in our considered scenarios. 

There are several interesting research directions for future multicast
ISAC networks. One particularly promising direction is expanding the
design to multi-group multicast channels, where the presence of inter-group
interference cannot be overlooked. Next, extending our work to a multi-cell
scenario is another interesting direction, in which different BSs
need to coordinate their signals effectively to mitigate multi-cell
interference. Moreover, investigating multicast ISAC with other sensing
tasks represents another interesting aspect, as it involves employing
diverse sensing metrics. Additionally, regarding to secure multicast
ISAC networks, taking into account secrecy capacity, holds potential
for future exploration. This specific area requires further investigation
to ensure robust and secure communications. \vspace{-0.2cm}

\appendices{}

\section{Proof of Proposition 1\vspace{-0.2cm}
}

Let $\{\mu_{k}\geq0\}$ and $\lambda\ge0$ denote the dual variables
associated with the constraints in \eqref{eq:rate constraint} and
\eqref{eq:Power constraint}, respectively. Then, the Lagrangian of
problem (P1.1) is\vspace{-0.2cm}
\begin{align}
\mathcal{L}(\boldsymbol{S}_{x},\lambda,\{\mu_{k}\})= & \mathrm{tr}(\boldsymbol{S}_{x}^{-1})+\Gamma\sum_{k=1}^{K}\mu_{k}+\mathrm{tr}\big((\lambda\boldsymbol{I}-\sum_{k=1}^{K}\mu_{k}\boldsymbol{h}_{k}\boldsymbol{h}_{k}^{H})\boldsymbol{S}_{x}\big)-\lambda P.
\end{align}
\vspace{-0.2cm}
Accordingly, the dual function is \vspace{-0.2cm}
\begin{eqnarray}
g\big(\lambda,\{\mu_{k}\}\big)= & \underset{\boldsymbol{S}_{x}\succeq\boldsymbol{0}}{\min} & \mathcal{L}(\boldsymbol{S}_{x},\lambda,\{\mu_{k}\}),\label{eq:dual function}
\end{eqnarray}
\vspace{-0.2cm}
for which we have the following lemma.\vspace{-0.2cm}

\begin{lem}
\textup{For $g\big(\lambda,\{\mu_{k}\}\big)$ to be bounded from below,
it must hold that $\boldsymbol{A}(\lambda,\{\mu_{k}\})\overset{\triangle}{=}\lambda\boldsymbol{I}-\sum_{k=1}^{K}\mu_{k}\boldsymbol{h}_{k}\boldsymbol{h}_{k}^{H}\succeq\boldsymbol{0}$.}\vspace{-0.2cm}
\end{lem}
\begin{proof} Suppose that $\boldsymbol{A}(\lambda,\{\mu_{k}\})\succeq\boldsymbol{0}$
does not hold. Let $\hat{\alpha}<0$ denote any one negative eigenvalue
of $\boldsymbol{A}(\lambda,\{\mu_{k}\})$ and $\hat{\boldsymbol{u}}$
denote its associated eigenvector. In this case, we can set $\boldsymbol{S}_{x}=\rho\hat{\boldsymbol{u}}\hat{\boldsymbol{u}}^{H}$
with $\rho\rightarrow\infty$ such that $\mathcal{L}(\boldsymbol{S}_{x},\lambda,\{\mu_{k}\})\rightarrow-\infty$,
i.e., we have $g\big(\lambda,\{\mu_{k}\}\big)\rightarrow-\infty$.
This thus contradicts the fact that $g\big(\lambda,\{\mu_{k}\}\big)$
is bounded from below. Therefore, $\boldsymbol{A}(\lambda,\{\mu_{k}\})\succeq\boldsymbol{0}$
must hold true. \end{proof} Based on Lemma 2, the dual problem of
(P1.1) is given by \vspace{-0.2cm}
\begin{subequations}
\begin{align}
\textrm{(D1.1)}:\underset{\lambda,\{\mu_{k}\}}{\max} & ~g\big(\lambda,\{\mu_{k}\}\big)\nonumber \\
\mathrm{s.t.} & ~\boldsymbol{A}(\lambda,\{\mu_{k}\})\succeq\boldsymbol{0},\label{eq:Amatrix}\\
 & \lambda\geq0,\mu_{k}\geq0,\forall k\in\mathcal{K}.\label{eq:nonnegative}
\end{align}
\end{subequations}
 For notational convenience, let $\mathcal{D}$ denote the feasible
region of $\lambda$ and $\{\mu_{k}\}$ characterized by \eqref{eq:Amatrix}
and \eqref{eq:nonnegative}. Let $\boldsymbol{S}_{x}^{*}$ denote
the optimal solution to problem \eqref{eq:dual function} with given
$(\lambda,\{\mu_{k}\})\in\mathcal{D}$. Furthermore, let $\boldsymbol{S}_{x}^{\mathrm{opt}}$
denote the optimal primal solution to problem (P1), and $\lambda^{\mathrm{opt}}$
and $\{\mu_{k}^{\mathrm{opt}}\}$ denote the optimal dual solution
to problem (D1.1).

As the strong duality holds between problem (P1.1) and its dual problem
(D1.1), problem (P1.1) can be solved by equivalently solving dual
problem (D1.1). In the following, we first solve problem \eqref{eq:dual function}
with a given set $(\lambda,\{\mu_{k}\})\in\mathcal{D}$, then find
$\lambda^{\mathrm{opt}}$ and $\{\mu_{k}^{\mathrm{opt}}\}$ for problem
(D1.1), and finally obtain the optimal primary solution $\boldsymbol{S}_{x}^{\mathrm{opt}}$
to problem (P1.1).

First, we evaluate the dual function $g\big(\lambda,\{\mu_{k}\}\big)$
with any given $(\lambda,\{\mu_{k}\})\in\mathcal{D}$. To this end,
we suppose that $\mathrm{rank}(\boldsymbol{A}(\lambda,\{\mu_{k}\}))=N\le N_{t}$
without loss of generality. Accordingly, we express the \ac{evd}
of $\boldsymbol{A}(\lambda,\{\mu_{k}\})$ as $\boldsymbol{A}(\lambda,\{\mu_{k}\})=\boldsymbol{U}\boldsymbol{\Lambda}\boldsymbol{U}^{H}$,
where $\boldsymbol{U}\boldsymbol{U}^{H}=\boldsymbol{U}^{H}\boldsymbol{U}=\boldsymbol{I}$,
and $\boldsymbol{\Lambda}=\mathrm{diag}(\alpha_{1},\dots,\alpha_{N_{t}})$
with $\alpha_{1}\ge\dots\ge\alpha_{N_{t}}$ being the eigenvalues
of $\boldsymbol{A}(\lambda,\{\mu_{k}\})$. We then have the following
Lemma\footnote{Note that Lemma 3 provides the optimal solution to problem (56) for
any given $(\lambda,\{\mu_{k}\})\in\mathcal{D}$, for which $\mathrm{rank}(\boldsymbol{A}(\lambda,\{\mu_{k}\})<N_{t}$
may hold such that the eigenvalues of $\boldsymbol{S}_{x}^{*}$ can
be unbounded. This is different from Proposition 1, which considers
the optimal dual variables $(\lambda^{\mathrm{opt}},\{\mu_{k}^{\mathrm{opt}}\})$,
for which $\mathrm{rank}(\boldsymbol{A}(\lambda^{\mathrm{opt}},\{\mu_{k}^{\mathrm{opt}}\}))=N_{t}$
follows.}.\vspace{-0.2cm}

\begin{lem}
\label{lemma1} \textup{\label{lem:1}For any given $(\lambda,\{\mu_{k}\})\in\mathcal{D}$,
we have the optimal solution $\boldsymbol{S}_{x}^{*}$ to problem
(\ref{eq:dual function}) as}\vspace{-0.2cm}
\textup{
\begin{equation}
\boldsymbol{S}_{x}^{*}=\boldsymbol{U}{\boldsymbol{\Sigma}}\boldsymbol{U}^{H},\label{eq:opt:Sx}
\end{equation}
where ${\boldsymbol{\Sigma}}=\mathrm{diag}(\tau_{1},\dots,\tau_{N_{t}})$
with $\tau_{i}=\alpha_{i}^{-1/2}$ for $i\leq N$ and $\tau_{i}\rightarrow+\infty$
for $N<i\leq N_{t}$.}\begin{proof}Please refer to Appendix B. \end{proof}
\end{lem}
Based on Lemma \ref{lemma1}, the dual function $g\big(\lambda,\{\mu_{k}\}\big)$
is obtained. 

Next, we find the optimal dual solution $\lambda^{\mathrm{opt}}$
and $\{\mu_{k}^{\mathrm{opt}}\}$ to solve dual problem (D1.1). As
problem (D1.1) is always convex but non-differentiable in general,
we adopt sub-gradient based methods such as the ellipsoid method to
find its optimal solution. The basic idea of the ellipsoid method
is to first generate an ellipsoid containing $\lambda^{\mathrm{opt}}$
and $\{\mu_{k}^{\mathrm{opt}}\}$, and then iteratively construct
new ellipsoids containing these variables with reduced volumes, until
convergence \cite{boyd2008ellipsoid}. To successfully implement the
ellipsoid method, we only need to find the sub-gradients of the objective
and constraint functions in problem (D1.1). 

First, we consider the objective function of problem (D1.1), one sub-gradient
of which at any given $[\mu_{1},\dots,\mu_{K},\lambda]^{T}\in\mathbb{C}^{(K+1)\times1}$
is $[\mathrm{tr}\big(\boldsymbol{H}_{1}\boldsymbol{S}_{x}^{*}\big)-\Gamma,\dots,\mathrm{tr}\big(\boldsymbol{H}_{K}\boldsymbol{S}_{x}^{*}\big)-\Gamma,P-\mathrm{tr}(\boldsymbol{S}_{x}^{*})]^{T}$.
Then, we consider the constraints in \eqref{eq:nonnegative}. It is
clear that the subgradient for constraint $\mu_{k}\geq0$ is $-\boldsymbol{e}_{k},\forall k\in\mathcal{K}$,
and the subgradient for $\lambda\geq0$ is $-\boldsymbol{e}_{K+1}$.
Finally, we consider constraint $\boldsymbol{A}(\lambda,\{\mu_{k}\})\succeq\boldsymbol{0}$
in \eqref{eq:Amatrix}, whose subgradient is given in the following
lemma. For notational convenience, we define $\boldsymbol{H}_{k}=\boldsymbol{h}_{k}\boldsymbol{h}_{k}^{H},\forall k\in\mathcal{K}$.\vspace{-0.2cm}

\begin{lem}
\textup{Let $\boldsymbol{v}\in\mathbb{C}^{N_{t}\times1}$ denote the
eigenvector of $\boldsymbol{A}(\lambda,\{\mu_{k}\})$ corresponding
to its smallest eigenvalue. The sub-gradient of }$\boldsymbol{A}(\lambda,\{\mu_{k}\})\succeq\boldsymbol{0}$\textup{
in \eqref{eq:Amatrix} at given }$\big(\lambda,\{\mu_{k}\}\big)\in\mathcal{D}$
\textup{is}\vspace{-0.2cm}
\textup{
\begin{equation}
\begin{array}{c}
[\boldsymbol{v}^{H}\boldsymbol{H}_{1}\boldsymbol{v},\dots,\boldsymbol{v}^{H}\boldsymbol{H}_{K}\boldsymbol{v},-1]^{T}.\end{array}
\end{equation}
}\begin{proof} This lemma follows by noting that $\boldsymbol{A}(\lambda,\{\mu_{k}\})\succeq\boldsymbol{0}\Longleftrightarrow\boldsymbol{v}^{H}\boldsymbol{A}(\lambda,\{\mu_{k}\})\boldsymbol{v}\geq0.$
\end{proof}
\end{lem}
With the sub-gradients of the objective function and all the constraint
functions obtained, we can efficiently implement the ellipsoid method
to obtain the optimal dual solution $\lambda^{\mathrm{opt}}$ and
$\{\mu_{k}^{\mathrm{opt}}\}$. 

Notice that at the optimal dual solution $\lambda^{\mathrm{opt}}$
and $\{\mu_{k}^{\mathrm{opt}}\}$, it must follow that $\lambda^{\mathrm{opt}}>0$
and $\boldsymbol{A}(\lambda^{\mathrm{opt}},\{\mu_{k}^{\mathrm{opt}}\})$
is of full rank (i.e., $\mathrm{rank}(\boldsymbol{A}(\lambda^{\mathrm{opt}},\{\mu_{k}^{\mathrm{opt}}\}))=N_{t}$),
since otherwise, the maximum transmit power constraint in \eqref{eq:Power constraint}
cannot be satisfied. Then, Proposition 1 follows directly from Lemma
3. This completes the proof.\vspace{-0.5cm}

\section{Proof of Lemma 3}

First, we have $\mathrm{tr}(\boldsymbol{A}(\lambda,\{\mu_{k}\})\boldsymbol{S}_{x})=\mathrm{tr}(\boldsymbol{\Lambda}\boldsymbol{U}^{H}\boldsymbol{S}_{x}\boldsymbol{U})$.
Let $\tilde{\boldsymbol{S}_{x}}=\boldsymbol{U}^{H}\boldsymbol{S}_{x}\boldsymbol{U}$.
It is easy to figure out that $\mathrm{tr}(\boldsymbol{S}_{x}^{-1})=\mathrm{tr}(\tilde{\boldsymbol{S}_{x}}^{-1})$.
Recall that $\boldsymbol{S}_{x}$ is positive semi-definite. We denote
$(\tau_{1},\dots,\tau_{N_{t}})$ as the diagonal entries of $\tilde{\boldsymbol{S}_{x}}$
to be determined. Note that $\mathrm{tr}(\boldsymbol{A}(\lambda,\{\mu_{k}\})\boldsymbol{S}_{x})=\sum_{i=1}^{N}\alpha_{i}\tau_{i}$.
Here, we introduce the following lemma to find the minimum of $\mathrm{tr}(\tilde{\boldsymbol{S}_{x}}^{-1})$
w.r.t. $(\tau_{1},\dots,\tau_{N_{t}})$, for which the proof can be
found in \cite[Appendix A]{ohno2004capacity} and thus is omitted.
\begin{lem}
\textup{\cite{ohno2004capacity} For a positive semi-definite matrix
$\boldsymbol{B}_{0}\in\mathbb{C}^{M\times M}$, with $(m,n)$-th entry
$a(m,n),$ it holds that}\vspace{-0.2cm}
\textup{
\begin{equation}
\mathrm{tr}(\boldsymbol{B}_{0}^{-1})\geq\sum_{i=1}^{M}\frac{1}{a(i,i)},
\end{equation}
}\vspace{-0.2cm}
\textup{where the equality holds if and only if $\boldsymbol{B}_{0}$
is diagonal.}\vspace{-0.2cm}
 
\end{lem}
Hence, $\tilde{\boldsymbol{S}_{x}}$ must be diagonal and we obtain
\vspace{-0.2cm}
\begin{eqnarray}
\mathrm{tr}(\boldsymbol{A}(\lambda,\{\mu_{k}\})\boldsymbol{S}_{x})+\mathrm{tr}(\boldsymbol{S}_{x}^{-1}) & = & \mathrm{tr}(\boldsymbol{\Lambda}\tilde{\boldsymbol{S}_{x}})+\mathrm{tr}(\tilde{\boldsymbol{S}_{x}}^{-1})=\sum_{i=1}^{N}\alpha_{i}\tau_{i}+\sum_{i=1}^{N_{t}}\frac{1}{\tau_{i}}.
\end{eqnarray}
\vspace{-0.2cm}
In this case, when $N<N_{t},$ for $i>N$, $\tau_{i}$ must approach
infinity to achieve the minimum value 0, i.e., for $i>N$, $\tau_{i}\rightarrow+\infty$.
As for $i\leq N$, by checking the first order differentiation, we
have $\tau_{i}=\alpha_{i}^{-1/2}$. Then, we obtain $\boldsymbol{S}_{x}^{*}$
as\vspace{-0.5cm}
\begin{equation}
\boldsymbol{S}_{x}^{*}=\boldsymbol{U}{\boldsymbol{\Sigma}}\boldsymbol{U}^{H},
\end{equation}
where ${\boldsymbol{\Sigma}}=\mathrm{diag}(\tau_{1},\dots,\tau_{N_{t}})$
with $\tau_{i}\rightarrow+\infty$, for $i>N$, and $\tau_{i}=\alpha_{i}^{-1/2},$
for $i\leq N$.

When $\boldsymbol{A}(\lambda,\{\mu_{k}\})$ is full-rank, we have\vspace{-0.5cm}
\begin{eqnarray}
\mathrm{tr}(\boldsymbol{A}(\lambda,\{\mu_{k}\})\boldsymbol{S}_{x})+\mathrm{tr}(\boldsymbol{S}_{x}^{-1}) & = & \mathrm{tr}(\boldsymbol{\Lambda}\tilde{\boldsymbol{S}_{x}})+\mathrm{tr}(\tilde{\boldsymbol{S}_{x}}^{-1})=\sum_{i=1}^{N_{t}}\alpha_{i}\tau_{i}+\sum_{i=1}^{N_{t}}\frac{1}{\tau_{i}}.
\end{eqnarray}
For any $i\leq N_{t}$, we need to set $\tau_{i}=\alpha_{i}^{-1/2}$
to achieve the minimum value. Then, the $\boldsymbol{S}_{x}^{*}$
is given by $\boldsymbol{S}_{x}^{*}=\boldsymbol{U}^{H}{\boldsymbol{\Sigma}}\boldsymbol{U},$
where $\tilde{\boldsymbol{\Lambda}}=\mathrm{diag}(\tau_{1},\dots,\tau_{N_{t}})$
with $\tau_{i}=\alpha_{i}^{-1/2},$ for $i\leq N_{t}$.

{\footnotesize{}{} \bibliographystyle{IEEEtran}
\bibliography{IEEEabrv,IEEEexample,my_ref}
}{\footnotesize\par}
\end{document}